\documentclass[11pt]{article} 
\pdfoutput=1
\usepackage{fullpage}
\usepackage{amssymb}
\usepackage{amsmath}
\usepackage{amsthm}
\usepackage{stmaryrd}
\usepackage[normalem]{ulem}
\usepackage{thmtools}
\usepackage{enumerate}
\usepackage{palatino}
\usepackage{complexity}
\usepackage{todonotes}
\usepackage{paralist}
\usepackage{microtype}
\usepackage{bbm}

\usepackage{color}
\usepackage[pdfstartview=FitH,pdfpagemode=None,colorlinks=true,citecolor=blue,linkcolor=blue]{hyperref}

\usepackage{bm}
\usepackage{xspace}

\numberwithin{equation}{section}
\declaretheoremstyle[bodyfont=\it,qed=$\triangle$]{noproofstyle}

\declaretheorem[numberlike=equation]{theorem}
\declaretheorem[numberlike=equation,name=Theorem,style=noproofstyle]{theoremwp}
\declaretheorem[numberlike=equation]{lemma}
\declaretheorem[numberlike=equation,name=Lemma,style=noproofstyle]{lemmawp}
\declaretheorem[numberlike=equation]{corollary}
\declaretheorem[numberlike=equation,name=Corollary,style=noproofstyle]{corollarywp}

\declaretheorem[unnumbered,name=Theorem,style=noproofstyle]{theorem*}
\declaretheorem[unnumbered,name=Lemma,style=noproofstyle]{lemma*}
\declaretheorem[unnumbered,name=Corollary,style=noproofstyle]{corollary*}
\declaretheorem[unnumbered,name=Proposition,style=noproofstyle]{proposition*}

\declaretheorem[unnumbered,name=Claim]{claim*}
\declaretheorem[unnumbered,name=Conjecture]{conjecture*}
\declaretheorem[unnumbered,name=Question]{question*}

\declaretheorem[name=Theorem,style=noproofstyle]{thmm}

\declaretheoremstyle[qed=$\lozenge$]{defstyle} 
\declaretheorem[numberlike=equation,style=defstyle]{definition}
\declaretheorem[unnumbered,name=Definition,style=defstyle]{definition*}

\declaretheorem[unnumbered,name=Example,style=defstyle]{example*}

\declaretheorem[unnumbered,name=Notation=defstyle]{notation*}
\declaretheorem[numberlike=equation,style=defstyle]{construction}
\declaretheorem[unnumbered,name=Construction,style=defstyle]{construction*}

\declaretheorem[numberlike=equation,style=defstyle]{remark}
\declaretheorem[unnumbered,name=Remark,style=defstyle]{remark*}

\newenvironment{innerproof}[1]{\proof}{\endproof}

\newenvironment{proof-sketch}{\medskip\noindent{\em Sketch of
Proof.}\hspace*{1em}}{\qed\bigskip}
\newenvironment{proof-attempt}{\medskip\noindent{\em Proof attempt.}\hspace*{1em}}{\bigskip}
\newenvironment{proof-of}[1]{\medskip\noindent\emph{Proof of #1.}\hspace*{1em}}{\qed\bigskip}

\newcommand{\inparen }[1]{\left(#1\right)}             
\newcommand{\inbrace }[1]{\left\{#1\right\}}

\DeclareMathOperator{\rank}{rank}

\DeclareMathOperator{\Span}{span}

\DeclareMathOperator{\supp}{supp}

\newcommand{\ceil}[1]{\lceil #1 \rceil}
\newcommand{\floor}[1]{\lfloor #1 \rfloor}

\newcommand{\set}[1]{\inbrace{#1}}
 
\newcommand{\ip}[2]{\left\langle{#1},{#2}\right\rangle}

\newcommand{\F}{\mathbb{F}}
\newcommand{\N}{\mathbb{N}}

\newcommand{\Z}{\mathbb{Z}}

\newcommand{\naive}{na\"{\i}ve\xspace}

\newcommand{\veca}{{\vec{a}}}
\newcommand{\vecaa}{{\vec{\alpha}}}
\newcommand{\vecb}{{\vec{b}}}
\newcommand{\vecbb}{{\vec{\beta}}}
\newcommand{\vecc}{{\vec{c}}}

\newcommand{\vecf}{{\vec{f}}}
\newcommand{\vecg}{{\vec{g}}}

\newcommand{\vecs}{{\vec{s}}}
\newcommand{\vect}{{\vec{t}}}

\newcommand{\vecv}{{\vec{v}}}

\newcommand{\vecx}{{\vec{x}}}
\newcommand{\vecy}{{\vec{y}}}
\newcommand{\vecz}{{\vec{z}}}

\newcommand{\GSV}[1]{\mathcal{G}^{\text{SV}}_{#1}}
\newcommand{\GKS}[1]{\mathcal{G}^{\text{KS}}_{#1}}
\newcommand{\cC}{\mathcal{C}}
\newcommand{\cG}{\mathcal{G}}
\newcommand{\cGH}{\mathcal{G}^\text{H}}
\newcommand{\GFS}{\mathcal{G}^\text{FS}}
\newcommand{\cGHFS}{\mathcal{G}^\text{H+FS}}
\newcommand{\cH}{\mathcal{H}}
\renewcommand{\O}{\mathcal{O}}
\newcommand{\bits}{\{0,1\}}
\newcommand{\ignore}[1]{}
\newcommand{\eqdef}{:=}

\DeclareMathOperator{\rspan}{row-span}
\DeclareMathOperator{\cspan}{col-span}

\renewcommand{\epsilon}{\varepsilon}
\renewcommand{\phi}{\varphi}

\renewcommand{\vec}[1]{\boldsymbol{#1}}

\DeclareMathOperator{\Coeff}{Coeff}

\newcommand{\coeff}[2]{\Coeff_{#1}\inparen{#2}}

\newcommand{\deriv}{\partial}

\newcommand{\ind}[1]{\mathbbm{1}_{#1}}
\newcommand{\zr}[1]{{\llbracket {#1}\rrbracket}}

\title{Pseudorandomness for Multilinear Read-Once Algebraic Branching Programs, in any Order}
\author{
Michael A.\ Forbes\thanks{Email: \texttt{miforbes@mit.edu},
Department of Electrical Engineering and Computer Science, MIT CSAIL, 32 Vassar St.,
Cambridge, MA 02139. This work supported by the Center for Science of Information (CSoI), an NSF Science and Technology Center, under grant agreement CCF-0939370.}
\and Ramprasad Saptharishi\thanks{Email: \texttt{ramprasad@cmi.ac.in}, Microsoft Research India, Bangalore, India.}
		\and
Amir Shpilka\thanks{Faculty of Computer Science, Technion --- Israel Institute of Technology, Haifa, Israel,
\texttt{shpilka@cs.technion.ac.il}.  The research leading to these results has received funding
from the European Community's Seventh Framework Programme (FP7/2007-2013) under grant agreement number 257575.}}

\begin{document}
\maketitle

\begin{abstract}
	\sloppy We give deterministic black-box polynomial identity testing algorithms for multilinear read-once oblivious algebraic branching programs (ROABPs), in $n^{\O(\lg^2 n)}$ time.\footnote{All logarithms are base 2\ in this paper, except where otherwise noted.}  Further, our algorithm is oblivious to the order of the variables.  This is the first sub-exponential time algorithm for this model. Furthermore, our result has no known analogue in the model of read-once oblivious \emph{boolean} branching programs with unknown order, as despite recent work (eg. \cite{BogdanovPW11,ImpagliazzoMZ12,ReingoldSV13}) there is no known pseudorandom generator for this model with sub-polynomial seed-length (for unbounded-width branching programs).

	This result extends and generalizes the result of Forbes and Shpilka~\cite{ForbesShpilka12a} that obtained a $n^{\O(\lg n)}$-time algorithm when given the order. We also extend and strengthen the work of Agrawal, Saha and Saxena~\cite{AgrawalSS12} that gave a black-box algorithm running in time $\exp((\lg n)^{\Omega(d)})$ for set-multilinear formulas of depth $d$. We note that the model of multilinear ROABPs contains the model of set-multilinear algebraic branching programs, which itself contains the model of set-multilinear formulas of arbitrary depth. We obtain our results by recasting, and improving upon, the ideas of Agrawal, Saha and Saxena~\cite{AgrawalSS12}.  We phrase the ideas in terms of \emph{rank condensers} and \emph{Wronskians}, and show that our results improve upon the classical multivariate Wronskian, which may be of independent interest.

	\fussy In addition, we give the first $n^{\O(\lg\lg n)}$ black-box polynomial identity testing algorithm for the so called model of diagonal circuits. This model, introduced by Saxena~\cite{Saxena08} has recently found applications in the work of Mulmuley~\cite{Mulmuley12}, as well as in the work of Gupta, Kamath, Kayal, Saptharishi~\cite{GuptaKKS13}. Previously, Agrawal, Saha and Saxena~\cite{AgrawalSS12}, Forbes and Shpilka~\cite{ForbesShpilka12a}, and Forbes and Shpilka~\cite{ForbesShpilka13} had given $n^{\Theta(\lg n)}$-time algorithms for this class. More generally, our result holds for any model computing polynomials whose partial derivatives (of all orders) span a low dimensional linear space. 
\end{abstract}

\section{Introduction}

Polynomial Identity Testing (PIT for short) is the problem of {\em deterministically} deciding whether a given algebraic circuit computes the identically zero polynomial. Recall that algebraic circuits are the algebraic analogs of boolean circuits that use the algebraic operations $\{+,\times\}$ to compute polynomials over an underlying field $\F$. The complexity of the problem comes from the requirement that the algorithm be deterministic. Indeed, efficient randomized algorithms follow easily from the Demillo-Lipton-Schwartz-Zippel Lemma~\cite{DemilloL78,Schwartz80,Zippel79}. The importance of the PIT problem stems from its many applications --- it is tightly connected to the problem of obtaining lower bounds for algebraic circuits (see for example, \cite{HeintzSchnorr80,KabanetsImpagliazzo04,Agrawal05,DSY09}) and it has applications in algorithms (see for example, \cite{KUW86,MVV87,AKS04}). For more on PIT see the survey of Shpilka and Yehudayoff~\cite{SY10}.  

PIT comes in two flavors. In the \emph{white-box} setting, the algorithm is given the circuit, namely, it is given the graph of computation. In the \emph{black-box} setting, the algorithm can access the circuit only through queries. A black-box PIT algorithm can be observed to be equivalent to a \emph{hitting set}, which is a small set of points so that any nonzero polynomial will evaluate to nonzero on one of these points. Thus, the black-box model allows weaker access, and thus algorithms in this model are correspondingly stronger results.  We seek algorithms in this restricted model, as they more closely match the performance of the randomized algorithm, which is black-box.

In recent years, there has been much work on PIT focused on restricted models of algebraic circuits. From the works of Agrawal and Vinay~\cite{AgrawalVinay08} and Gupta, Kamath, Kayal and Saptharishi~\cite{GuptaKKS13} it follows that solving PIT for depth-$3$ circuits would yield a solution to the case of general algebraic circuits. This has increased the focus on restricted models, and yielded a long line of works focused on PIT for depth-$3$ and depth-$4$ circuits \cite{DvirShpilka06, KayalSaxena07, KarninShpilka08,KayalSaraf09,SaxenaSeshadhri11,KMSV10,SarafVolkovich11,SahaSS11}.

\sloppy Another active line of research has considered models of computation with other types of restrictions. One line considers models that are \emph{set-multilinear}, in particular focusing on \emph{set-multilinear algebraic branching programs (ABPs)} \cite{RazShpilka05,ArvindMS08,ForbesShpilka12,ForbesShpilka12a,AgrawalSS12,ForbesShpilka13}. One reason for the focus on this model is that it captures the {\em partial derivative} technique that is the main technique today for proving lower bounds for algebraic circuits. Another reason for the interest shown in this model is that it is the algebraic analog for the branching program model in boolean complexity. In fact, the PIT problem for \emph{read-once oblivious ABPs (ROABPs)} can be viewed as the algebraic analog of the famous \ComplexityFont{L} vs. \RL\ question. Indeed, the PIT algorithm of Forbes and Shpilka~\cite{ForbesShpilka12a} for ROABPs is very close in nature to Nisan's~\cite{Nisan92} pseudo-random generator (PRG) for read-once oblivious (boolean) branching programs, which are a non-uniform version of randomized log-space Turing machines. For more on ROABPs we refer the reader to the introduction of the paper Forbes-Shpilka~\cite{ForbesShpilka12a}.

\fussy The work of Forbes and Shpilka~\cite{ForbesShpilka12a} highlighted a difference between known results for the boolean case and the algebraic case. While over the boolean domain we do not have an improvement to Nisan's PRG when the branching program is of bounded width, in the algebraic setting Forbes and Shpilka~\cite{ForbesShpilka12a} give an $n^{\lg n/\lg\lg n}$ PIT algorithm for the case of bounded width ROABPs, which, intuitively, would correspond to a PRG with seed-length $\O(\lg^2n/\lg\lg n)$. Obtaining such a seed length in the boolean setting is a long-standing open problem (see \cite[Open Problem 8.6]{Vadhan12}).  

In this work, we obtain new PIT algorithms for ROABPs that, again, do not yet have boolean analogues.  Previous work on read-once branching programs, for both the boolean and algebraic models, only worked for a fixed variable order.  That is, the branching program is read-once and reads each variable in a fixed \emph{known} order.  If this order is known, then black-box PIT algorithms, or PRG constructions, exist in previous work.  However, once the order becomes unknown, much less is known.  However, recently Agrawal, Saha, and Saxena~\cite{AgrawalSS12} introduced the idea of \emph{shifting} and \emph{rank concentration}, and using these ideas gave black-box PIT algorithms for sub-models of unknown-order ROABPs.  In this work, we extend and reinterpret their ideas, and consequently construct a quasipolynomial time black-box PIT algorithm for (multilinear) ROABPs when the order of the variables is unknown. In contrast, in the boolean setting the best algorithm when the order of the variables is unknown runs in time $\exp(\sqrt{n})$, due to Impagliazzo, Meka and Zuckerman~\cite{ImpagliazzoMZ12}.

\subsection{The models}\label{sec:the model}

In this work, we consider various restricted models of computation, which we now define.  They are all variants, or sub-models, of algebraic branching programs (ABPs).  Algebraic branching programs were first defined in the work of Nisan~\cite{Nisan91} who proved exponential lower bounds on the size of non-commutative ABPs computing the determinant or permanent polynomials. We shall consider several variants of the basic model.

\begin{definition}[Nisan~\cite{Nisan91}]\label{def: ABP}
	An \emph{algebraic branching program with unrestricted weights} of depth $D$ and width $\le r$, on the variables $x_1,\ldots,x_n$, is a directed acyclic graph such that
	\begin{itemize}

		\item The vertices are partitioned in $D+1$ layers $V_0,\ldots,V_D$, so that $V_0=\{s\}$ ($s$ is the source node), and $V_D=\{t\}$ ($t$ is the sink node). Further, each edge goes from $V_{i-1}$ to $V_{i}$ for some $0< i\le D$.
		\item $\max|V_i|\le r$.
		\item Each edge $e$ is weighted with a polynomial $f_e\in\F[\vecx]$.

	\end{itemize}

	Each $s$-$t$ path is said to compute the polynomial which is the product of the labels of its edges, and the algebraic branching program itself computes the sum over all $s$-$t$ paths of such polynomials.
	\begin{itemize}

		\item In an \emph{algebraic branching program (ABP)}, for each edge $e$ the weight $f_e(\vecx)$ is an affine function.  The \emph{size} is $nDr$.
		\item In a \emph{read-once oblivious ABP (ROABP)} of (individual) degree $<d$, we have $n=D$, and there is some permutation $\pi:[n]\to [n]$ such that for each edge $e$ from $V_{i-1}$ to $V_i$, the weight is a univariate polynomial $f_e(x_{\pi(i)})\in\F[x_{\pi(i)}]$ of degree $<d$. The \emph{size} is $ndr$.  We say the \emph{variable order is known} if this permutation is the identity permutation, and it is \emph{unknown} otherwise.

		\item A polynomial is computed by a \emph{width$-r$ commutative ROABP}, if it is computable by a width-$r$ ROABP in \emph{every} variable order.

		\item In a \emph{set-multilinear algebraic branching program (SMABP)} the $dn$ variables $\{x_1,\ldots,x_{dn}\}$ are partitioned into $d$ sets of size $n$, $\vecx=\vecx_1\sqcup \cdots \sqcup \vecx_d$. For each edge $e$ from $V_{i-1}$ to $V_i$, $f_e$ is a (homogeneous) linear function in the variable set $X_i$. We say the \emph{partition is known} if  $\vecx_i = \{x_{(i-1)n+1},\ldots,x_{(i-1)n+n}\}$  and it is \emph{unknown} otherwise.\qedhere

	\end{itemize}
\end{definition}

\sloppy Note that a polynomial computed by a commutative ROABP is trivially computed by a ROABP. This class of polynomials coincide with the notion of polynomials with low \emph{evaluation dimension}, a notion of Saptharishi, as reported in Forbes-Shpilka~\cite{ForbesShpilka12a}.

\sloppy The definition of a SMABP is in line with the term coined by Nisan and Wigderson~\cite{NisanWigderson96} that defined a set-multilinear monomial to be a multilinear monomial that contains exactly one variable from each $\vecx_i$ and a set-multilinear polynomial as a polynomial consisting of set-multilinear monomials.  Thus, it is clear from the definition that a set-multilinear ABP computes a set-multilinear polynomial.  It is also not hard to see that any set-multilinear polynomial can be computed by a set-multilinear ABP, but perhaps requiring large size.

\fussy Another model that we consider in this paper is the so called {\em diagonal circuit} model, first defined by Saxena~\cite{Saxena08}. While this model is somewhat less motivated than small depth circuits or ABPs, it has found recent applications in the work of Mulmuley~\cite{Mulmuley12} concerning derandomization of Noether's Normalization Lemma, see also Forbes-Shpilka~\cite{ForbesShpilka13} for a discussion.  Furthermore, this model and the tools used to understand it were used in the recent work of Gupta-Kamath-Kayal-Saptharishi~\cite{GuptaKKS13}.  Lastly, this model is the weakest model known for which polynomial-time black-box PIT algorithms are not currently known. 

\begin{definition}[Diagonal circuits]\label{def: diagonal}
	A \emph{diagonal circuit} in the $n$ variables $\vecx$, of degree $d$, and size $s$, computes polynomials $f(\vecx)$ of degree at most $d$ that can be expressed as a sum of $s$ powers of linear functions. Namely, $f(\vecx)=\sum_{i=1}^{s}L_i(\vecx)^{d_i}$, where each $L_i$ is a linear function and $d_i\leq d$. 
\end{definition}

Diagonal circuits are examples of polynomials where the space of partial derivatives is low-dimensional.  We present the formal definition of this notion in \autoref{sec:hashing+FS}.

\subsection{Our results}\label{sec:our results}

We now give a summary of our results, and start by defining the notion of a hitting set, which is the object we seek to construct.

\begin{definition}[Hitting Set]
	Let $\cC$ be a class of  $n$-variate polynomials, with coefficients in $\F$.   A set $\cH\subseteq \F^n$ is a \emph{hitting set for $\cC$} if for all $f\in\cC$, $f\equiv 0$ iff $f|_\cH\equiv 0$.

	The hitting set $\cH$ is \emph{$t(n)$-explicit} if given an index into $\cH$, the corresponding element of $\cH$ can be computed in $t(n)$-time.
\end{definition}

Our main result is an explicit quasipolynomial-size hitting set for multilinear ROABPs, for unknown variable order.  As a special case (see \autoref{lem:smabp-to-roabp}), we obtain the same parameter hitting sets for set-multilinear ABPs, when the partition to the sets is unknown.

\begin{thmm}[Hitting set for multilinear unknown order ROABPs, informal version of \autoref{hit-set-noncommut}]\label{mainthm: non-commutative}
	Let $\cC$ be the set of $n$-variate multilinear polynomials computable by a width-$r$ ROABP in any variable order.  If $|\F|\ge \poly(n,r)$, then $\mathcal{C}$ has a $\poly(n,r)$-explicit hitting set, of size $\le \poly(n,r)^{\O(\lg r\cdot \lg n)}$.
\end{thmm}

For polynomials with individual degree $d$ (so each variable has degree $d$), the exponent in the hitting set increases by a factor of $d$.  To achieve this result, we used the ideas of Agrawal, Saha and Saxena~\cite{AgrawalSS12}, who gave results for sub-models of unknown-order ROABPs. As we reinterpret their ideas, we rederive some of their results. One such result we give below.

\begin{thmm}[Hitting set for commutative ROABPs, informal version of \autoref{hit-set-commut}]\label{mainthm: q-p commutative}
	Let $\cC$ be the set of $n$-variate polynomials with individual degree $<d$ that are computable by a width-$r$ commutative ROABPs (and thus computable by a ROABP in \emph{every} variable order).  If $|\F|\ge \poly(n,d)$, then $\mathcal{C}$ has a $\poly(n,d,r)$-explicit hitting set, of size $\le \poly(n,d)^{\O(\lg r)}$.
\end{thmm}

The proofs of the above results, while yielding sub-exponential-size hitting sets, do not readily suggest a path to obtaining polynomial-size hitting sets.  In particular, at first glance it seems difficult to see how the above proof strategies could yield hitting sets of size $n^{o(\lg n)}$.  However, our next result shows that for certain sub-models of commutative ROABPs, such as diagonal circuits, we can in fact achieve $n^{\O(\lg\lg n)}$-size hitting sets.

\begin{thmm}[Informal version of \autoref{nlglgn}]\label{mainthm: commutative}
	Let $\cC$ be the set of $n$-variate polynomials of degree $d$ computable by size $s$ diagonal circuits. If $|\F|\ge \poly(n,d,s)$, then $\mathcal{C}$ has a $\poly(n,d,s)$-explicit hitting set, which has size $\le \poly(n,d,s)^{\O(\lg \lg (ds))}$.

	Let $\cC$ be the set of multilinear $n$-variate polynomials computable by width-$r$ commutative ROABPs. If $|\F|\ge \poly(n,r)$, then $\mathcal{C}$ has a $\poly(n,r)$-explicit hitting set, which has size $\le \poly(n,r)^{\O(\lg \lg r)}$.
\end{thmm}

We also show that the above holds more generally for polynomials with a small space of partial derivatives, see \autoref{sec:hashing+FS}.  We can also extend the above result, by relaxing the multilinearity of the commutative ROABP to individual degree $d$, in which case the exponent grows by a factor of $\lg d$. In the previous theorem, the resulting exponent was independent of $d$, but that exponent has a worse dependence on $r$.

The proof of our results combines various works.  We start by extending and reinterpreting the ideas of shifting and rank concentration, of Agrawal, Saha and Saxena~\cite{AgrawalSS12}.  We then build on these techniques by incorporating new recursion schemes, as well as existing tools in the literature such as the Klivans and Spielman~\cite{KlivansSpielman01} generator for sparse polynomials, the generator of Shpilka and Volkovich~\cite{ShpilkaVolkovich09} for read-once formulas,  the generator of Forbes and Shpilka~\cite{ForbesShpilka12a} for ROABPs, as well as perfect hash functions.

\subsection{Related work}\label{sec: related}

\paragraph{Boolean Pseudorandomness:} The PIT problem can be seen as the algebraic analog of the task of constructing pseudorandom generators for boolean computation.  That is,  for a class $\cC$ of boolean circuits, we can seek to construct a \emph{pseudorandom generator (PRG)} $\cG:\bits^s\to\bits^n$ for $\cC$, such that for any circuit $C\in\cC$ on $n$ inputs, we have the $\epsilon$-closeness of distributions $C(\cG(U_s))\approx_\epsilon C(U_n)$, where $U_k$ denotes the uniform distribution on $\bits^k$. Nisan~\cite{Nisan92} studied pseudorandom generators for space-bounded computation, and for space $s$ computation gave a generator with seed length $s=\O(\lg^2 s)$. Randomized space-bounded computation can be modeled with read-once oblivious (boolean) branching programs, and these generators apply to this model of computation as well. Note that in the context of space-bounded derandomization, the order that the randomness is read is known, as this information is dictated when simulating the space-bounded machine.  

Despite Nisan's result, the seed length has not been improved for general read-once branching programs, despite much work on the problem (see for example \cite{Nisan92,INW94,RazReingold99,NisanZuckerman96,BogdanovDVY13,SimaZ11,GopalanMRTV12,BravermanRRY10,BrodyV10,KouckyNP11,De11,Steinke12,Tzur09,BogdanovPW11, ReingoldSV13, ImpagliazzoMZ12} among others).  While past PRG work for read-once branching programs has examined restricted branching programs (such as \emph{regular} or \emph{permutation}), more recently it has been observed that existing PRGs depend heavily on the order the variables are read, and Tzur~\cite{Tzur09} even showed that certain instantiations provably fail for adversarially chosen read orders.  Starting with the work of Bogdanov, Papakonstantinou and Wan~\cite{BogdanovPW11}, there have been attempts to obtain PRGs for read-once branching programs in any variable order.  In recent works, Impagliazzo, Meka and Zuckerman~\cite{ImpagliazzoMZ12} gave a generator with seed length $n^{1/2 + o(1)}$ for general branching programs in any variable order, and Reingold, Steinke and Vadhan~\cite{ReingoldSV13} gave a seed-length of $\O(\lg^2 n)$ for constant-width permutation read-once branching programs in any variable order.

In comparison, the work of Forbes and Shpilka~\cite{ForbesShpilka12a} studied read-once oblivious (algebraic) branching programs, and achieved a quasi-polynomial-sized hitting set, which would correspond to a ``seed length'' of $\O(\lg^2 s)$ for ABPs of size $s$. This result is analogous (in some sense) to the generator of Nisan for space bounded computations. Correspondingly, it only applies when the variable order of the ROABP is known. Independently, Agrawal, Saha and Saxena~\cite{AgrawalSS12} gave results for sub-models of ROABPs, and in their models it is somewhat less natural to assume that the variable order is known.  

Thus, motivated by the analogous questions in boolean pseudorandomness, as well as by the positive results for the unknown order case by Agrawal, Saha and Saxena~\cite{AgrawalSS12}, this work studies general ROABPs when the order of variables is unknown. We give a quasi-polynomial-sized hitting set, corresponding to a seed of length $\O(\lg^3 s)$ for ABPs of size $s$. It is an interesting question as to whether any insights of this paper (and other recent works) can be used to obtain a similar seed length in the boolean regime, or in general whether there are any formal connections between algebraic and boolean pseudorandomness for read-once oblivious branching programs.

\paragraph{The Partial Derivative Method:} \sloppy The PIT algorithm of Raz and Shpilka~\cite{RazShpilka05} works for any model of computation that outputs polynomials whose space of partial derivatives is (relatively) low-dimensional.\footnote{More accurately, it works for the set-multilinear model with a given ordering of the sets and we only consider derivatives according to variables from consecutive sets. That is, in the ABP model we consider only derivatives according to variables in adjacent layers.} Set-multilinear ABPs, non-commutative ABPs, low-rank tensors, and so called pure algebraic circuits (defined in Nisan and Wigderson~\cite{NisanWigderson96}) are all examples of algebraic models that compute polynomials with that property, and so the algorithm of Raz and Shpilka~\cite{RazShpilka05} works for them. In some sense using information on the dimension of partial derivatives of a given polynomial is the most applicable technique when trying to prove lower bounds for algebraic circuits (see e.g., \cite{Nisan91,NisanWigderson96,Raz06,Raz09a,RSY08,RazYehudayoff09,GuptaKKS12}) and so it was an interesting problem to understand whether this powerful technique could be carried over to the black-box setting. The work Forbes and Shpilka~\cite{ForbesShpilka12a} achieved a quasi-polynomial hitting set for this model when the order of variables is known in advance. Here we remove this restriction thus providing a black-box analog of Raz and Shpilka~\cite{RazShpilka05} (for multilinear polynomials), albeit with a quasi-polynomial running time.

\fussy Furthermore, when the polynomial has the property that the set of all of its partial derivatives spans a low-dimensional space (i.e., of polynomial dimension), then we get a nearly polynomial time algorithm, namely, an $n^{\O(\lg \lg n)}$ algorithm, as stated in \autoref{sec:hashing+FS}.

\paragraph{Previous Work:} As mentioned before, the work Forbes and Shpilka~\cite{ForbesShpilka12a} gave a hitting set of quasi-polynomial size for the class of ROABPs when the order of variables is known in advance. The basic idea in that work, which follows the idea of the prior work of Forbes and Shpilka~\cite{ForbesShpilka12}, is that of preserving dimension of linear spaces while reducing the number of variables. To achieve this reduction, Forbes and Shpilka~\cite{ForbesShpilka12a} used a variant of a constrution of Gabizon and Raz \cite{GabizonRaz08} which in this work we refer to as a \emph{rank condenser}. This is then used in a recursive merge-and-reduce scheme. In each step, the variables in the ROABP are partitioned into adjacent pairs according to the variable order, and these variables are merged (in parallel) in a brute-force manner. This halves the number of variables, but in a certain sense squares the degree of the polynomial. The rank condenser then allows these new variables to be sampled pseudorandomly (and in parallel), such that the degree is then reduced appropriately.  Thus, the number of variables is halved, and no other parameters are changed.  However, the pseudorandom sampling requires a ``seed'' of length $\O(\lg s)$, and there are $\lg n$ such seeds, giving a total of $\O(\lg n \lg s)$ seed length. Note that this process crucially uses the variable order, as it needs to know which variables are neighbors and thus can be merged together.

Recently, Agrawal, Saha and Saxena~\cite{AgrawalSS12} obtained results on black-box derandomization of PIT for small-depth set-multilinear formulas (as well as generalizations to when the multilinearity condition is relaxed), when the partition of the variables into sets is unknown. They obtained a hitting set of size $\exp((2h^2\lg(s))^{h+1})$, for size $s$ formulas of multiplicative-depth $h$.\footnote{The multiplicative-depth is the maximal number of product gates in an input-output path. In the bounded depth model this is essentially half the total depth.} In the language of pure formulas, they consider pure formulas of multiplicative-depth $h$. We note that this model is a strict subset of the more general model of pure formulas which itself is a strict sub-model of set-multilinear ABP.\footnote{The usual transformation from formulas to ABPs can be done preserving set-multilinearity.}  Thus, our result significantly improves upon the multilinear results of Agrawal-Saha-Saxena~\cite{AgrawalSS12}, and in particular the size of our hitting set does not depend on the depth of the formula.

\sloppy In this work we also present black-box PIT results for the model of diagonal circuits, or more generally, polynomials with a low dimensional space of partial derivatives.  Saxena~\cite{Saxena08} defined this model as a way capture some of the complexity of depth-4 circuits. He gave a polynomial-time white-box PIT algorithm, by  using algorithm of Raz and Shpilka\cite{RazShpilka05}.  Later, Saha, Saptharishi and Saxena~\cite{SahaSS11}, generalized Saxena's  results to the so-called semi-diagonal model. Simultaneously and independently (using very different techniques), Agrawal, Saha and Saxena~\cite{AgrawalSS12}, and Forbes and Shpilka~\cite{ForbesShpilka12a} gave $s^{\O(\lg s)}$ black-box PIT algorithms for size $s$ semi-diagonal depth-4 circuits. Later, Forbes and Shpilka~\cite{ForbesShpilka13} gave another such algorithm, that was somewhat simpler.  In this work, we give a $s^{\O(\lg\lg s)}$-size hitting set for diagonal circuits, and the techniques can be seen as merging the ideas of Agrawal-Saha-Saxena~\cite{AgrawalSS12}, Forbes-Shpilka~\cite{ForbesShpilka12a} and Forbes-Shpilka~\cite{ForbesShpilka13}.

\fussy

\subsection{Proof technique}

The work of Agrawal, Saha and Saxena~\cite{AgrawalSS12} introduced a new technique which they called {\em rank-concentration}. By shifting the variables $\vecx$, to $\vecx+\vect$ for a careful choice of $\vect$, they showed that for depth-$3$ set-multilinear formulas, they can reduce to the case where the partial derivative space is spanned by derivatives of small order, which allowed them to conclude that the shifted polynomial contains a low degree monomial. They then used induction on the depth to obtain rank-concentration for higher depth formulas as well. Once this is achieved it is relatively easy to find a nonzero substitution to the polynomial. The main difficulty comes from finding a good shift that guarantees the rank-concentration property. 

In this work we present the idea of Agrawal-Saha-Saxena~\cite{AgrawalSS12} in the language of rank-condensers, and of which the notion of a \emph{Wronskian} is a special case. Specifically, we show that shifts act on the space of partial derivatives as a rank condenser (see \autoref{sec:shifts}).  This can be phrased as a Wronskian-result, as it says that the rank of certain matrices of partial derivatives capture linear independence in multivariate polynomials. Using this terminology we are able to easily compute good shifts. We believe that the language of rank condensers offers a convenient language for these results, and thus do not use various notions from Agrawal-Saha-Saxena~\cite{AgrawalSS12} such as Hadamard algebras.

We now discuss the notion of a rank condenser (and defer the discussion of Wronskians to \autoref{sec:wronskian}), and how shifting can be viewed this way. Let $f_1(\vecx),\dots, f_r(\vecx)$ be $n$-variate polynomials. We can express these polynomials as a matrix $M$ whose rows are indexed by monomials $\vecx^\veca$ and whose columns are indexed by $i\in[r]$, such that $M_{\veca,i}$ is the coefficient of $\vecx^\veca$ in $f_i$. Note that this captures all of the information about the $f_i$.  In particular, these polynomials are zero iff the matrix $M$ is zero.  More generally, any linear combination of these polynomials is zero iff the corresponding linear combination of columns of $M$ is zero. We can test if such a linear combination of the columns $M$ is zero, by reading a subset of the rows, in particular, it suffices to examine a basis for the row-span of $M$.  However, from an algebraic perspective, we do not have unfettered query access to $M$.  Rather, some rows (corresponding to monomials) are easier to access than others.  In particular, those monomials of low degree, or even those of small support, are easier to examine via brute force.  Thus, the condition we would like on $M$ is that it has a row-basis among small-support monomials.  In particular, we say the $f_i$ have \emph{support-$\ell$ rank concentration} if $M$ has such a basis among support-$\ell$ monomials.  If such rank concentration occurs, it follows that a linear combination of the $f_i$ is nonzero iff it has a nonzero monomial of support-$\ell$.  Such a monomial can then be found via brute-force.

However, as the monomial $x_1x_2\cdots x_n$ shows, small-support rank concentration does not happen always.  However, if we perturb this monomial to $\prod_i (x_i+1)$, then it suddenly does have small-support rank concentration.  Thus, we seek to find a small set $\cH$ such that some shift $\vecx\to\vecx+\vecaa$ for some $\vecaa\in\cH$ induces small-support rank-concentration.  Taking a different perspective, we can ask for a shift $\vecx\to\vecx+\vect$, where $\vect$ is an algebraically simple polynomial map.

To show how to construct such shifts $\vect$, we need to study how the shifting process affects the matrix $M$.  First, we see that this action is linear, and thus there is some \emph{transfer matrix} $T(\vect)$ describing its action.  As shifting by $\vect$ then shifting by $-\vect$ gives the identity map, it follows that $T(\vect)$ is invertible.  But we want something more, that shifting by $\vect$ induces rank concentration.  In particular, it follows that shifting to small-support rank concentration means that we want \[\rank(M)= \rank ((T(\vect)^\text{[small]}) M),\] where $T(\vect)^\text{[small]}$ is the set of rows in $T$ corresponding to small-support monomials.  Thus, while $T$ is square, $T(\vect)^\text{[small]}$ is a short, fat matrix.  Yet, we still want it to preserve the rank of $M$ after multiplication.  Such matrices we call rank condensers. One primary tool in understanding whether $T$ is a rank condenser is the following.

\begin{lemma*}[Informal version of \autoref{lem:rank-condenser-recipe}]
	There is a generic recipe to turn any good error-correcting code into a good rank condenser.
\end{lemma*}

Thus, to show $T(\vect)^\text{[small]}$ is a rank condenser, we simply show it has some error-correcting code properties.  This establishes that shifting will condense rank, for generic $\vect$.  However, by itself this is not helpful as $\vect$ has too many variables. In the case of a commutative ROABP, as argued in Agrawal-Saha-Saxena~\cite{AgrawalSS12}, it can be shown that $\vect$ need only be a good shift for all small sets of variables, in which case $\vect$ can be found in quasipolynomial time.

To extend the rank concentration idea to general ROABPs and set-multilinear ABPs, we follow the merge-and-reduce idea of Forbes and Shpilka~\cite{ForbesShpilka12a}.  However, the rank condenser used in that work seemingly can only be implemented if the variable order is known.  However, we observe that the shifting idea of Agrawal-Saha-Saxena~\cite{AgrawalSS12} is actually a rank condenser, and is insensitive to variable order.  However, to implement this, we must create a variant of the shifting rank condenser, that only seeks to condense rank from a ``known basis''.  That is, if we seek to condense \emph{all} rank, then this will be too expensive.  Instead, we shift several times.  In each shift, we progressively make the rank more concentrated, and we save on seed length by only condensing rank from where we already concentrated it.

To establish the hitting sets of size $s^{\O(\lg\lg s)}$ for sub-models of commutative ROABPS, we need to combine a variety of tools in this area.  The rank concentration idea shows that we can shift commutative ROABPs such that they only have $\lg s$ essential variables, where this is possible for multilinear commutative ROABPs because of the Klivans-Spielman~\cite{KlivansSpielman01} generator.  We then use hashing to reduce the problem to one on these essential variables.  However, we will need to do this in a manner compatible with the read-once nature of the computation, and thus use the Shpilka-Volkovich~\cite{ShpilkaVolkovich09} generator.  After these steps, we have a commutative ROABP in $\lg s$ variables, and the Forbes-Shpilka~\cite{ForbesShpilka12a} generator, while quasipolynomial, has a very good dependence on the number of variables of the ROABP.  Thus, by composing with this generator, we then reduce to $\lg\lg s$ variables, at which point we can interpolate to get a $s^{\O(\lg\lg s)}$-size hitting set.

\subsection{Organization}

The paper is organized as follows. In \autoref{sec:tools} we give some basic tools and definitions that we shall use throughout the paper. In \autoref{sec:shifts} we discuss shifts and rank condensers. \autoref{sec:commutative} contains the proof of \autoref{mainthm: q-p commutative}. In \autoref{sec:condense-known-basis} and \autoref{sec:noncommutative} we prove \autoref{mainthm: non-commutative}. Finally, in \autoref{sec:hashing+FS} we prove \autoref{mainthm: commutative}.

\subsection{Notation and definitions}\label{sec:notation}

For a polynomial $f\in\F[\vecx]$, we shall use $\coeff{\vecx^\veca}{f}$ to denote the coefficient of the monomial $\vecx^\veca$ in $f$. We will sometimes identify the monomial $\vecx^\veca$ with its exponent vector $\veca$. We shall often encounter \emph{Hasse derivatives}, where we define $\deriv_{\vecx^{\vecb}}(f)(\vect)\eqdef\coeff{\vecx^{\vecb}}{f(\vecx + \vect)}$.  We will use $\deriv_{\vecx^{\vecb}}(f)$ for the derivative evaluated at the original variables $\vecx$.  As the derivative is linear, we will also view the derivative $\deriv_{\vecx^{\vecb}}$ as a linear operator acting on $\F[\vecx]$, and will also consider linear combinations of these operators, which we refer to as differential operators. We use Hasse derivatives as they work well in fields of arbitrary characteristic, but have most of the usual properties of normal derivatives in zero characteristic.  For more, see Dvir-Kopparty-Saraf-Sudan~\cite{DvirKSS09} or Forbes-Shpilka~\cite{ForbesShpilka13}.

We will use the notation $[n]$ to denote $\{1,\ldots,n\}$, and $\zr{n}$ to denote $\{0,\ldots,n-1\}$.  We will occasionally index vectors and matrices from zero, and will use $\zr{n}$ to denote this, as $\F^{\zr{n}\times\zr{m}}$ will denote an $n\times m$ matrix, indexed from zero. For a set $S$ and an integer $k$ we denote by ${S \choose k}$ the set of all subsets of $S$ of size $k$.

We will often work with vectors and matrices, and they will often contain entries that are polynomials. Operations on polynomials, such as extracting coefficients, or taking derivatives, will be extended to vectors and matrices coordinate-wise. Vectors, such as $\vecf(\vecx)\in\F[\vecx]^r$, will be written in bold, where as matrices, such as $F(\vecx)\in\F[\vecx]^{n\times m}$ will be written with capital letters. We will also sometimes identify matrices of polynomials in $\F[\vecx]^{n\times m}$ with vectors of polynomials in $\F[\vecx]^{nm}$. Given an $n\times m$ matrix $M$ and sets $S\subseteq[n]$ and $T\subseteq[m]$, we use the notation $M|_{S\times T}$ to denote the submatrix of $M$ with rows from $S$ and columns from $T$.  When we do not restrict the rows (respectively, the columns) of $A$, we will use the symbol `$\bullet$' in place of $S$ (respectively, $T$).  When the sets $S$ and $T$ are singletons, we will use the more standard notation of $M_{i,j}$ to refer to the $(i,j)$-th entry of $M$. Given matrices $M_i$, their non-commutative product $\prod_{i\in[n]} M_i$ will have the natural order associated, that is, the product $M_1\cdots M_n$. Similarly, for a vector $v\in \F^n$ and $S\subseteq[n]$ we denote by $v|_S$ the restriction of $v$ to the coordinates of $S$. In a similar fashion, when we have two sets of variables $(\vecx,\vecy)$ and a vector $v$ of length $|\vecx|+|\vecy|$ we denote with $v|_\vecx$ the restriction of $v$ to the coordinates associated with the variables in $\vecx$, and similarly define $v|_\vecy$.

Similar to the above abuse of notation, we will extend the notion of a ROABP to compute matrices.  That is, we will say a matrix $M(\vecx)\in\F[\vecx]^{r\times r}$ is computable by a width-$r$ ROABP if $M(\vecx)=\prod_{i\in[n]} M_i(x_i)$ for $M_i(x_i)\in\F[x_i]^{r\times r}$.

For a vector $\veca\in\N^n$, we denote $|\veca|_0=|\{i\mid a_i \neq 0\}|$, $|\veca|_1\eqdef a_1+\cdots+a_n$, $|\veca|_\infty=\max_i a_i$, and $\vec{a}\,!=a_1!\cdots a_n!$. Given a monomial $\vecx^\veca$, we will say it has (total) degree $|\veca|_1$, individual degree $|\veca|_\infty$, and support $|\veca|_0$.  We will define the support of a monomial, denoted $\supp(\vecx^\veca)$, to be either the coordinates $\{i: a_i\ne 0\}$, or the corresponding variables $\{x_i:a_i\ne 0\}$, and this will depend on context. For any $a < b$, we shall set $\binom{a}{b}$ to be zero. Also, for vectors $\veca = \inparen{a_1,\dots, a_n}$ and $\vecb = \inparen{b_1, \dots, b_n}$, we shall define $\binom{\veca}{\vecb} = \prod_i \binom{a_i}{b_i}$. 

For a set of vectors $T$ we denote with $\dim_{\F}(T)$ the dimension of the vector space spanned by the vectors in $T$, over $\F$. For a polynomial $f(\vecx)$ and a variable $x_i$ we denote the degree of $f$ in $x_i$ by $\deg_{x_i}(f)$.

\section{Tools}\label{sec:tools}

In this section we state some known results about ABPs and state some results from the literature that will be used throughout the paper. 

\paragraph{ABPs and Iterated Matrix Multiplication:}
We start by giving the well known equivalence of ABPs and iterated matrix multiplication (see, e.g., Lemma 3.1 of Forbes-Shpilka~\cite{ForbesShpilka12a}). 

\begin{lemmawp}
	\label{lem:roabp-as-imm}
	Let $f\in\F[\vecx]$ be computed by a depth $D$, width $\le r$ ABP with unrestricted weights, such that the variable layers are $V_0,\ldots,V_D$.  For $0<i\le D$, define $M_i\in\F[\vecx]^{V_{i-1}\times V_i}$ such that the $(u,v)$-th entry in $M_i$ is the label on the edge from $u\in V_{i-1}$ to $v\in V_i$, or 0 if no such edge exists. Then,  $f(\vecx)$ is a linear combination of the entries in $\prod_{i\in[d]} M_i(\vecx)\eqdef M_1(\vecx) M_2(\vecx)\cdots M_D(\vecx)$.
	
	Further, for an ABP, the matrix $M_i$ has entries that are affine forms, and for a ROABP of individual degree $<d$, the matrix $M_i$ has entries which are univariate polynomials in $x_{\pi(i)}$ of degree $<d$ (where $\pi$ is the permutation determining the order of variables).
	 For a set-multilinear ABP on the variable sets $\vecx=\vecx_1 \sqcup \vecx_2 \sqcup \cdots \sqcup \vecx_d$ the matrix $M_i$ has entries that are linear forms in the variables $\vecx_i$. 
\end{lemmawp}

Given the above, we will focus on iterated matrix multiplication, in particular for matrices that are all $r\times r$.  Further, as it is more amenable to recursion, we will more focus on the matrix product $\prod_i M_i(\vecx)$ than the polynomial $f(\vecx)$.

\paragraph{Set-Multilinear and ROABPs:} We note here that set-multilinear ABPs can be expressed as multilinear ROABPs of slightly larger width, and thus will focus on ROABPs.

\begin{lemma}\label{lem:smabp-to-roabp}
	Let $f(\vecx_1,\ldots,\vecx_d)$ be computed by a set-multilinear ABP of width-$r$.  Then $f$ can be computed by a multilinear ROABP of width-$2r$ in the $dn$ variables, in some variable order.
\end{lemma}
\begin{proof}
	Consider first a single homogeneous linear form $L$ in some partition $\vecx_i$, so that $L(\vecx_i)=\sum_{j\in[n]} a_{i,j} x_{i,j}$.  Note that 
	\[
		\begin{bmatrix}
			1	&	L(\vecx_i)\\
			0	&	1\\
		\end{bmatrix}
		=
		\prod_{j\in[n]}
		\begin{bmatrix}
			1	&	a_{i,j}x_{i,j}\\
			0	&	1\\
		\end{bmatrix}.
	\]
	Now consider any $r\times r$ matrix $M_i(\vecx_i)$ in a set-multilinear ABP.  By using the above transformation, and appealing to block-matrix multiplication, we can construct width-$2r$ matrices $M_{i,j}$ so that $\prod_{j\in[n]} M'_{i,j}(x_{i,j})$ embeds $M_i(\vecx_i)$.  Further, as this embedding is regular, it follows there are $2r\times 2r$ matrices $P$ and $Q$ such that $P\cdot \prod_{j\in[n]} M'_{i,j}(x_{i,j})\cdot Q=M_i(\vecx_i)$.  Incorporating $P$ and $Q$ into the product, shows that $M_i(\vecx_i)$ can be computed by the requisite ROABP.  Applying this argument to each partition $\vecx_i$ and multiplying the resulting matrices gives the result.
\end{proof}

\paragraph{PIT results:}

The first PIT result that we need is a generator of Klivans and Spielman~\cite{KlivansSpielman01} for sparse polynomials. We will not need to hit sparse polynomials themselves, but rather will need the hashing properties of this generator that allow us to reduce to fewer variables.

\begin{theoremwp}[Klivans-Spielman generator~\cite{KlivansSpielman01}]\label{thm:KS generator}
	Let $|\F|\ge\poly(n,d)$, where $m=\Theta(\log_{nd} s)$. Then there is a $\poly(n,d,m)$-explicit polynomial $\GKS{n,d,s}:\F^m\times\F^m\to\F^n$ such that
	\begin{itemize}
		\item For all $i\in[n]$, the polynomial $(\GKS{n,d,s})_i$ has individual degree $\le \poly(n,d)$.
		\item For every set $S\subseteq\F[\vecx]$ of $n$-variate monomials of individual degree $<d$, such that $|S|\le s$, there is some $\vecaa\in\F^m$ such the polynomials $\{(\GKS{n,d,s}(\vect,\vecaa))^\veca\}_{\vecx^\veca\in S}$ are nonzero, distinct monomials in $\vect$.\qedhere
	\end{itemize}
\end{theoremwp}

The above is a slight variation of the work of Klivans and Spielman \cite{KlivansSpielman01}, and it can be constructed using their techniques.  That is, take $p$ a prime of size $\poly(n,d)$, and consider the substitution $x_i\leftarrow t^{k^{i}\pmod p}$, so that $\vecx^\veca\leftarrow t^{\veca(k)\pmod p}$ (at least, when examining the substitution modulo $t^p$), where $\veca(k)=\sum a_ik^i$ can be regarded as a polynomial of degree $\le n$.  Thus, it follows that two monomials $\vecx^\veca$ and $\vecx^\vecb$ agree under this substitution with probability $\le n/p=1/\poly(n,d)$ for a random $k\in\Z_p$.  By using the $\vecaa$ to interpolate through these choices of $k$, this yields the $m=1$ construction. For $m>1$, one takes $x_i\leftarrow t_1^{k_1^{i}\pmod p}\cdots t_m^{k_m^{i}\pmod p}$, and observes that two monomials then agree with probability $\le 1/\poly(n,d)^m$.  Taking a union bound over $s^2$ possible collisions amongst the monomials in $S$ shows that some value of the $\{k_j\}_j$, and thus of $\vecaa\in\F^m$, will yield distinct monomials in $\vect$.

Another important tool is the generator of Shpilka-Volkovich~\cite{ShpilkaVolkovich09}.

\begin{theoremwp}[Shpilka-Volkovich generator~\cite{ShpilkaVolkovich09}]\label{thm:SV generator}
	\sloppy Let $\F$ be a field of size $>n$. Let $\xi_0,\xi_1,\ldots,\xi_n$ be distinct elements in $\F$, and $\ell\in[n]$. Then there is a $\poly(n,\ell)$-explicit polynomial map $\GSV{n,\ell}(y_1,\ldots,y_\ell,z_1,\ldots,z_\ell)\in \F[\vecy,\vecz]^n$, such that
	\[
		(\GSV{n,\ell}(\vecy,\xi_{i_1},\ldots,\xi_{i_\ell}))_k
		=\sum_{j:i_j=k} y_j,
	\]
	for $0\le i_j\le n$, and $k\in[n]$. Specifically, $(\GSV{n,\ell})_k=\sum_{j\in[\ell]} \ind{z_j=\xi_k}\cdot y_j$, where $\ind{z=\xi_k}$ is a univariate polynomial in $z$ with degree $\le n$. 
\end{theoremwp}

That is, the $\ind{z=\xi_k}$ are the Lagrange interpolation polynomials, so that $\ind{z_j=\xi_k}= \delta_{j,k}$.  Note that this generator is in a sense $\ell$-wise independent.  That is, by setting the $z_j$ appropriately, we can insert the $\vecy$ variables into any $\ell$ of the $\vecx$ variables. Note that the setting $z_j\leftarrow \xi_0$ zeroes out the appropriate $y_j$ in the polynomial map (a convenience for handling $<\ell$-wise positions), where as $z_j\leftarrow \xi_i$ for $i\ge 1$ sends $y_j$ to the $x_i$ slot. It follows if a polynomial $f$ depends on only $\ell$ of the variables in $\vecx$ then $f\not\equiv 0$ iff $f\circ \GSV{n,\ell}\not\equiv0$. However, $\GSV{n,\ell}$ also has the property that it hits polynomials with small-support, see \autoref{lem:GSV-hit-ss}, which is more relevant for this paper.

We shall need some more tools for later theorems, and for the sake of clarity we shall state them as and when we need.

\section{Shifts, Rank Condensers, and Rank Concentration}\label{sec:shifts}

In this section we will explore the properties of the shifting map $f(\vecx)\mapsto f(\vecx+\vect)$.  Following Agrawal-Saha-Saxena~\cite{AgrawalSS12} and Forbes-Shpilka~\cite{ForbesShpilka12a}, we will look not just at one polynomial at a time, but rather at a set of polynomials, as these polynomials will represent intermediate computations in the computation. Given a vector of polynomials $\vecf = (f_1,\dots, f_r)\in\F[\vecx]^r$, we associate to it a matrix $M$ such that $M_{\veca,i}$ is the coefficient of $\vecx^{\veca}$ in the polynomial $f_i$.  Thus, $M$ will have $r$ columns.  The rows of $M$ are drawn from all monomials in $\vecx$ of low degree, where the degree bounds will depend on the context. 

The above discussion uses the language of coefficients of monomials.  This is equivalent to a notion of a Hasse derivative, the derivative language is sometimes more natural and thus will be used below.  Thus, we will think of the matrix $M$ with $M_{\veca,i}=\deriv_{\vecx^{\veca}}(f_i)(\vec{0})$.  Thus, the row of $M$ associated with $\veca$ is the $r$-dimensional vector $\deriv_{\vecx^{\veca}}(\vecf)(\vec{0})$.

Given this matrix, we can then study its rank, as well as its row span.  Clearly, the rank is $\le r$ by construction.  Further, the polynomials are all zero iff the rank is $0$.  If the polynomials are nonzero, then there is some basis of the row span of $M$, which yield derivatives $\deriv_{\vecx^{\veca}}(\vecf)(\vec{0})$ that are nonzero vectors.  If we can access these derivatives efficiently, this will allow us to decide if $\vecf$ is zero or not.  To capture this, we follow Agrawal-Saha-Saxena~\cite{AgrawalSS12} and define the notion of rank concentration.

\begin{definition}[Rank concentration]\label{def:rank-concentration}
	For a vector of polynomials $\vecf\in\F[\vecx]^r$, and a set of monomials $S\subseteq\F[\vecx]$, we shall say that $\vecf$ is \emph{rank concentrated on $S$ at $\vecaa\in\F^n$} if the derivatives of $\vecf$ with respect to the monomials in $S$ span all of the derivatives of $\vecf$. That is,
	\[
		\Span \set{\deriv_{\vecx^{\veca}}(\vecf)(\vecaa)}_{\veca \in S}\quad
		= \quad \Span \set{\deriv_{\vecx^{\veca}}(\vecf)(\vecaa)}_\veca.
	\]
	In particular, if $\vecf$ is rank concentrated on the set of monomials of support size at most $\ell$, then we shall say that $\vecf$ is \emph{support-$\ell$ rank concentrated}. 
\end{definition}

In this paper, we will frequently abuse this notation and apply it to matrices of polynomials $F\in\F[\vecx]^{n\times m}$, where we identify $\F[\vecx]^{n\times m}\equiv \F[\vecx]^{nm}$. In the above definition we have two vector spaces to consider, the span of derivatives in $S$, and the span of all derivatives. Note that the second vector space does not depend on the underlying point $\vecaa$, as stated in the following lemma, which follows from polynomial interpolation and the fact that Hasse derivatives are exactly the coefficients of the polynomial.

\begin{lemmawp}\label{change-deriv-point}
	Let $\vecf\in\F[\vecx]^r$ be a vector of polynomials, each of degree $\le d$, and suppose $|\F|> d$.  Then for any $\vecaa\in\F^n$, $\Span\{\deriv_{\vecx^{\veca}}(\vecf)(\vecaa)\}_{\veca}=\Span\{\vecf(\vecbb)\}_{\vecbb\in\F^n}$.
\end{lemmawp}

The goal of Agrawal-Saha-Saxena~\cite{AgrawalSS12}, as well as this work, is to construct small sets of $\vecaa$'s such that $\vecf$ must have small-support rank concentration at one such $\vecaa$.  The next two lemmas show why this suffices for polynomial identity testing. The first lemma shows that small-support rank concentration implies the existence of a small-support monomial with nonzero coefficient.

\begin{lemma}
	Let $\vecf\in\F[\vecx]^r$ be a vector of polynomials that is support-$\ell$ rank concentrated at $\vecaa\in\F^n$. Let $g(\vecx)\eqdef\ip{\vecbb}{\vecf}\eqdef\sum_{i\in[r]} \beta_i f_i$, for $\vecbb\in\F^r$.  Then $g(\vecx)$ is nonzero iff there is a small-support monomial $\vecx^\veca$ with $|\veca|_0\le\ell$ such that $\coeff{\vecx^\veca}{h}\ne 0$, where $h(\vecx)\eqdef g(\vecx+\vecaa)$.
\end{lemma}
\begin{proof}
	First, observe that $h\not\equiv 0$ iff $g\not\equiv 0$.

	\uline{$\impliedby$:} Clearly if $\coeff{\vecx^\veca}{h}\ne 0$, then $h\not\equiv 0$, so $g\not\equiv 0$.

	\uline{$\implies$:} That $g\not\equiv 0$, and thus $h\not\equiv 0$ implies there is some $\vecb$ such that 
	\begin{align*}
		0
		&\ne\coeff{\vecx^\vecb}{h}\\
		&=\coeff{\vecx^\vecb}{\ip{\vecbb}{\vecf(\vecx+\vecaa)}}\\
		&=\ip{\vecbb}{\coeff{\vecx^\vecb}{\vecf(\vecx+\vecaa)}}\\
		&=\ip{\vecbb}{\deriv_{\vecx^\vecb}(\vecf)(\vecaa)}.
	\end{align*}
	By the rank concentration we have that $\deriv_{\vecx^\vecb}(\vecf)(\vecaa)\in\Span\{\deriv_{\vecx^\veca}(\vecf)(\vecaa)\}_{|\veca|_0\le\ell}$, from which it follows that there is some $\veca$ with $|\veca|_0\le\ell$ such that $\ip{\vecbb}{\deriv_{\vecx^\veca}(\vecf)(\vecaa)}$ is nonzero.  By identical logic as the above, this inner-product equals $\coeff{\vecx^\veca}{h}$, yielding the claim.
\end{proof}

Thus, polynomial identity testing for polynomials with rank concentration reduces to the case when the polynomials have nonzero small support monomial.  We now show how to use the generator of Shpilka-Volkovich~\cite{ShpilkaVolkovich09} to hit such polynomials. Note that the ability to hit such polynomials does not follow from the usual \naive probabilistic method proof, as while such polynomials can be consider ``simple'' they are too numerous to use a \naive union bound. Note that Agrawal-Saha-Saxena~\cite{AgrawalSS12} did not use the Shpilka-Volkovich~\cite{ShpilkaVolkovich09} generator to hit polynomials with a small-support monomial, and instead used sufficiently many inputs of low-hamming-weight.  Indeed, that the SV-generator works in the below lemma can be proven by observing that this generator contains such low-hamming-weight strings in its image. However, the generator perspective will be crucial for our improvement from $n^{\O(\lg n)}$ to $n^{\O(\lg\lg n)}$ for hitting sets for diagonal circuits, as given in \autoref{sec:hashing+FS}.

\begin{lemma}\label{lem:GSV-hit-ss}
	Let $\F$ be a field with $|\F|>n$. Let $f\in\F[\vecx]$ be such that $f$ is nonzero iff there is some $\veca$ with $|\veca|_0\le \ell$ such that $\coeff{\vecx^\veca}{f}\ne 0$.  Then $f\not\equiv 0$ iff $f\circ\GSV{n,\ell}\not\equiv 0$.
\end{lemma}
\begin{proof}
	\uline{$\impliedby$:} This is clear.

	\uline{$\implies$:} Consider the monomial $\vecx^\veca$ of support size $\le\ell$ in $f$ with nonzero coefficient. Let $\supp(\veca)=\{x_{i_1},\ldots,x_{i_k}\}$ be its support for $k\le\ell$.  In $\GSV{n,\ell}$ of \autoref{thm:SV generator}, take $z_j=\xi_{i_j}$ for $j\in[k]$.  For $j\in[\ell]\setminus[k]$, take $z_j=\xi_0$.  It follows then that the variables outside $\supp(\veca)$ will receive zero from $\GSV{n,\ell}$, and inside $\supp(\veca)$ they will receive a distinct variable from $\vecy$. It follows then the monomials containing support outside $\supp(\veca)$ are zeroed out, and those inside are preserved.  As the monomial $\vecx^\veca$ has a nonzero coefficient, it follows that $f\circ\GSV{n,\ell}\not\equiv 0$ with this partial setting of the $\vecy,\vecz$, and thus is also true without any setting of $\vecy,\vecz$.
\end{proof}

Putting these two lemmas together, we get the following corollary.

\begin{corollarywp}[Rank concentration to Hitting sets]\label{cor:rank-conc-to-hitting-sets}
	Let $\vecf\in\F[\vecx]^r$ be a vector of polynomials that is support-$\ell$ rank concentrated at $\vecaa\in\F^n$. Let $g(\vecx)\eqdef\ip{\vecbb}{\vecf}\eqdef\sum_{i\in[r]} \beta_i f_i$, for $\vecbb\in\F^r$.  Then $g(\vecx)\not\equiv 0$ iff $g\circ (\GSV{n,\ell}+\vecaa)\not\equiv 0$.

	Furthermore, $g\circ(\GSV{n,\ell}+\vecaa)$ depends on $2\ell$ variables and has degree $\le n\deg(h)$. Hence such polynomials $g$ have a hitting set of size $\poly(n,\deg(h))^{2\ell}$.
\end{corollarywp}

We now explore conditions on the points $\vecaa$ that guarantee small-support rank concentration.  We begin with an example.

\subsection{Rank Concentration as a Wronskian}\label{sec:wronskian}

To illustrate how to achieve rank concentration, consider the case when the polynomials in $\vecf$ are univariate of degree $\le d$.  In this case, we will seek that the rank is concentrated on low-degree derivatives (as all derivatives in this setting are trivially support-1).  Specifically, we will want rank to be concentrated on the first $r$ derivatives. Let us write the vectors of all degree derivatives $\deriv_{x^i}(\vecf)$ at the point $t\in\F$ in matrix form, yielding
\[
	\begin{bmatrix}
		f_1(t) 				& \cdots & f_r(t)\\
		(\deriv_{x}f_1)(t)		& \cdots & (\deriv_x f_r)(t)\\ 
		\vdots 				& \ddots & \vdots \\
		(\deriv_{x^{d}} f_1)(t) 	& \cdots & (\deriv_{x^{d}} f_r)(t)
	\end{bmatrix}.
\]
That the rank is concentrated on low-degree means that the row-span of the first $r$ rows, denoted $V$, equals the span of all of the rows, denoted $W$.  As $V\subseteq W$, this will happen iff $\rank V=\rank W$.  However, the columns of the above matrix are the coefficients of the polynomials $f_i$ when expanded around $t$, and thus the column-rank of this matrix equals the rank of of the polynomials $\inbrace{f_1,\dots, f_r}$ as vectors in the vector space $\F[x]$.  Invoking row-rank equalling column-rank, it follows that the rank is concentrated at $t$ in the first $r$ rows iff the row-rank of these rows equals the rank of the $f_i$ in $\F[x]$.

The first $r$ rows of the above matrix are known as the \emph{Wronskian} of $\inbrace{f_1,\dots, f_r}$, and it is classically known (for fields of polynomially large characteristic) that the rank (over $\F(t)$) of this matrix equals the rank of the polynomials $\inbrace{f_1,\dots, f_r}$ as vectors in the vector space $\F[x]$ (see, for example, Bostan-Dumas~\cite{BoDu10} for a proof). Thus, by expressing this rank condition via determinants, it follows that there is a degree $\le rd$ polynomial $p(t)$ such that if $p(a)\ne 0$, then $\vecf$ has degree $<r$ rank concentration at $a$.  Thus, at least in this example, we see that rank concentration happens almost everywhere, and where it fails is a simple algebraic condition we can avoid by brute-force.

In this work, we will need to generalize the above Wronskian in various ways.  First, we will have multivariate polynomials.  A multivariate Wronskian is known (again, see Bostan-Dumas~\cite{BoDu10}), but in our application it only guarantees low-degree rank concentration, while this paper needs small-support rank concentration.  Further, much of the existing literature on Wronskians focus on fields of (polynomially) large characteristic, while we seek results over any characteristic. We note that several recent results in algebraic complexity (for example \cite{ChattopadhyayGKPS13,KayalS12}) have been limited to characteristic zero because of their dependence on the classical Wronskian.

\sloppy Recently Guruswami and Kopparty~\cite{GuruswamiK13} explored the notion of a \emph{subspace design}, and observed that they are in a sense dual to Wronskians. To obtain results over any characteristic, they introduced what they call a \emph{folded} Wronskian, in reference to the folded Reed-Solomon codes of Guruswami-Rudra~\cite{GuruswamiR08}.  They noted that this folded Wronskian was used in Forbes-Shpilka~\cite{ForbesShpilka12} and Guruswami-Wang~\cite{GuruswamiW13} in different guises.

\fussy

In this work, we extend existing knowledge of Wronskians.  First, we observe that the underlying algebraic object in the above is the notion of a \emph{rank condenser}, which we define in \autoref{sec:rankcondenser from code}, where we give a general recipe for constructing rank condensers from error-correcting codes.  That the (multivariate or univariate) Wronskian matrix characterizes the rank of polynomials can be expressed as saying that at certain derivative-related matrix is a good rank condenser.  To prove this fact, we will show that this derivative-related matrix has code-like properties.  Further, we observe that the folded Wronskian of Guruswami-Kopparty~\cite{GuruswamiK13} can interpreted as a being a Wronskian using a different notion of derivative. Specifically, the Wronskian uses the $q$-derivative, a $q$-analogue of the usual derivative, see Kac and Cheung~\cite{KacCheung}.  We note that Bostan, Salvy, Chowdhury, Schost and Lebreton~\cite{BostanSCSL12} observed the relation between folded Reed-Solomon codes and $q$-derivatives.

In this version of the paper, we will develop the tools needed to achieve small-support rank concentration for multivariate polynomials over any characteristic, and this will only use Hasse derivatives.  This can be easily interpreted as establishing a ``small-support multivariate Wronskian'' result, that works over any characteristic. In a future version of this work, we will give a multivariate rank condenser (that is a Wronskian) that simultaneously guarantees small-support and low-degree rank concentration, over any characteristic, appealing to $q$-derivatives in nonzero characteristic. This will strictly improve the classically known multivariate Wronskian, which may be of independent interest (for example, to the above works using Wronskians).

In the next sections, we will define the notion of a rank condenser, and give a construction based on the parity check matrix of a good error-correcting code.  We shall then show that small-support derivatives have such a code structure, which will give the needed results for small-support rank concentration.

\subsection{Rank Condensers from Good Codes}\label{sec:rankcondenser from code}

In this section we define rank condensers, and show how they can be generically constructed from error correcting codes, following the Cauchy-Binet proof strategy of Dvir, Gabizon and Wigderson~\cite{DvirGW09} (Dvir, Gabizon and Wigderson~\cite{DvirGW09} also had a notion of a rank extractor, but with a different notion of rank) and Forbes and Shpilka~\cite{ForbesShpilka12}.  We begin with the definition of rank condensers (which, given the parameters, could be considered as \emph{lossless}).

\begin{definition}[Rank Condenser Generator]
	A matrix $E(\vect) \in \F[\vect]^{n\times m}$ is said to be a \emph{(seeded) rank condenser (generator) for rank $r$} if for every matrix $M \in \F^{m\times k}$ of rank $\le r$, we have that $\rank_{\F(\vect)}(E(\vect)\cdot M)=\rank_\F(M)$. 
\end{definition}

Note that in the above definition, the notion of rank changes, as one rank is over the field of rational functions $\F(\vect)$, and the other rank is over $\F$.  This will be the setting we work in, but it is sometimes more convenient to think in terms of the base field $\F$, leading to the following definition.

\begin{definition}[Rank Condenser Hitting Set]
	A collection $\mathcal{E}\subseteq \F^{n\times m}$ is said to be a \emph{(seeded) rank condenser (hitting set) for rank $r$} if for every matrix $M \in \F^{m\times k}$ of rank $\le r$, there is some $E\in\mathcal{E}$ such that $\rank_{\F}(E\cdot M)=\rank_\F(M)$. 
\end{definition}

\begin{remark}
	A more natural definition of a (linear) condenser would be a family of linear maps $\phi_i:\F^m\to\F^n$, such that for any subspace $V\subseteq \F^m$ of rank $\ge r$, there is some $i$ such that $\rank\phi(V)=r$. It can be readily seen that this is equivalent to the hitting set definition above.  If $n=r$, then one can consider this an extractor.
\end{remark}

\sloppy By interpolation over $\vect$, the generator notion yields the hitting set.  Conversely, one can construct a generator by interpolating a curve through the collection $\mathcal{E}$.  See Shpilka and Volkovich~\cite{ShpilkaVolkovich09} for more on generators versus hitting sets within polynomial identity testing.

\fussy

In what follows, we will omit the terms ``seeded'', ``generator'' and ``hitting set'' when it is clear from context. Before turning to the construction, we need one more definition, that of a \emph{monomial order} (see Cox, Little and O'Shea~\cite{CLO} for more on monomial orderings).

\begin{definition}
	A \emph{monomial ordering} is a total order $\prec$ on the nonzero monomials in $\F[\vecx]$ such that
	\begin{itemize}
		\item For all $\veca\in\N^n$, $1\prec\vecx^{\veca}$.
		\item For all $\veca,\vecb,\vecc\in\N^n$, $\vecx^{\veca}\prec\vecx^{\vecb}$ implies $\vecx^{\veca+\vecc}\prec\vecx^{\vecb+\vecc}$
	\qedhere
	\end{itemize}
\end{definition}

For concreteness, one can consider the lexicographic ordering on monomials, which is easily seen to be a monomial ordering. We now give our first rank condenser construction.

\begin{lemma}\label{lem:rank-condenser-recipe}
	Let $E(\vect)\in\F[\vect]^{n\times m}$, with $E=\Lambda(\vect)^{-1}\cdot H\cdot W(\vect)$ such that
	\begin{itemize}
		\item $\Lambda(\vect)\in\F[\vect]^{n\times n}$ is a diagonal matrix, whose diagonal entries are nonzero monomials.
		\item $H\in\F^{n\times m}$ is the parity check matrix of an error correcting code of distance $>r$.  That is, every $r$ columns of $H$ are linearly independent.
		\item $W(\vect)\in\F[\vect]^{m\times m}$ is a diagonal matrix, whose diagonal entries are nonzero \emph{distinct} monomials.
	\end{itemize}
	Then $E(\vect)$ is a rank condenser for rank $r$. 
	
	Furthermore, let $d_i^{\max} =\max_{S\in\binom{[m]}{r}} \sum_{j\in S} \deg_{t_i} W(\vect)_{j,j}$ and $d_i^{\min}=\min_{S\in\binom{[m]}{r}} \sum_{j\in S} \deg_{t_i} W(\vect)_{j,j}$. Let $C_i\subseteq\F\setminus\{0\}$ have size $|C_i|>d_i^{\max}-d_i^{\min}$.  Then $\mathcal{E}=\{E(\vect)\}_{\forall i, t_i\in C_i}$ is a rank condenser for rank $r$.
\end{lemma} 
\begin{proof}
	Consider any matrix $M\in\F^{m\times k}$ of rank $s\le r$.

	\sloppy \uline{$k>s$:} We show the claim by reducing to the case to when $k=s$, which we prove next.  Let $M'\in\F^{m\times s}$ be a submatrix of $M$ formed by taking $s$ linearly independent columns of $M$.  It follows then that $\rank_\F(M')=\rank_\F(M)\ge	\rank_{\F(\vect)}(E(\vect)M)\ge\rank_{\F(\vect)}(E(\vect)M')=\rank_\F(M')$, where we use that row-rank equals column-rank.  More specifically, the first inequality uses that $\rspan_{\F(\vect)}(E(\vect)M)\subseteq \rspan_{\F(\vect)}(M)$ and that $\rank_\F(M)=\rank_{\F(\vect)}(M)$, and the second inequality uses that $\cspan_{\F(\vect)}(E(\vect)M')\subseteq\cspan_{\F(\vect)}(E(\vect)M)$.  From this, we see that all inequalities must be met with equality, so $\rank_\F(M)=\rank_{\F(\vect)}(E(\vect)M)$.

	\uline{$k=s$:} As $W$ is diagonal with monomial entries, it induces $w:[m]\to\F[\vect]$ weighting the columns of $H$, indexed by $[m]$, with monomials in $\F[\vect]$ (thus, $w(i)=W_{i,i}$). Let $\prec$ be some monomial order on $\F[\vect]$. As the monomials in $W$ are distinct, the order $\prec$ also defines a total order on $[m]$, which we will also call `$\prec$', abusing notation.

	As $M\in \F^{m\times s}$ is rank $s$, there are sets $T\subseteq[m]$ of size $s$ such that $\det(M|_{T\times\bullet})\ne 0$.  As the weights in $w$ are distinct by hypothesis on $W$, standard greedy/matroid arguments (for example, as seen in \autoref{lem:matroidgreedypartial} with $J=\emptyset$) imply that there is a unique $T_0\subseteq[m]$ of size $s$ such that $\det(M|_{T_0\times\bullet})\ne 0$, and that minimizes $w(T_0)$ in the $\prec$-order, where $w(T)\eqdef\prod_{i\in T} w(i)$.  Further, as every $s$ columns of $H$ are linearly independent, it follows that there an $S\subseteq[n]$ of size $s$ such that $\det(H|_{S\times T_0})\ne0$. We now use this set to show that the rank is extracted, via the following claim.

	\begin{claim*}
		$\det((E(\vect)M)|_{S\times\bullet})=\det((\Lambda(\vect)^{-1}HW(\vect)M)|_{S\times\bullet})\ne 0$.
	\end{claim*}
	\begin{innerproof}{Claim}
		As $\Lambda(\vect)$ is diagonal with nonzero diagonal entries, and the determinant is multiplicative, we see that
		\[\det((\Lambda(\vect)^{-1}HW(\vect)M)|_{S\times\bullet})=\det(\Lambda(\vect)|_{S\times S})^{-1}\det((HW(\vect)M)|_{S\times\bullet}).\]
		\sloppy As $E(\vect)$ is a matrix with polynomial entries, it follows that $\det((E(\vect)M)|_{S\times\bullet})$ is a polynomial, so $\det(\Lambda(\vect)|_{S\times S})$ is nonzero and divides $\det((HW(\vect)M)|_{S\times\bullet})$.  Thus, it suffices to show that $\det((HW(\vect)M)|_{S\times\bullet})=\det((HW(\vect))|_{S\times\bullet}\cdot M)$ is not zero. To analyze this determinant, we invoke the Cauchy-Binet formula (for example, see Forbes and Shpilka~\cite{ForbesShpilka12} for a proof), which we quote below.

		\fussy

		\begin{lemma*}[Cauchy-Binet Formula]
			Let $m\ge n\ge 1$.  Let $A\in\F^{n\times m}$, $B\in\F^{m\times n}$.  Then $\det(AB)=\sum_{T\in {[m] \choose n}} \det(A|_{\bullet\times T})\det(B|_{T\times\bullet})$
		\end{lemma*}

		Applying this formula, we see that
		\begin{align*}
			\det((HW(\vect))|_{S\times\bullet}\cdot M)
			&=\sum_{T\in {[m] \choose s}} \det((HW(\vect))|_{S\times T})\det(M|_{T\times\bullet})
			\intertext{as $W$ is diagonal and the determinant is multiplicative,}
			&=\sum_{T\in {[m] \choose s}} \det(H|_{S\times T})\det(M|_{T\times\bullet})\prod_{i\in T} W(\vect)_{i,i}\\
			&=\det(H|_{S\times T_0})\det(M|_{T_0\times\bullet})w(T_0)+\sum_{T\in {[m] \choose s}} \det(H|_{S\times T})\det(M|_{T\times\bullet})w(T)
			\intertext{using the choice of $T_0$ and the uniqueness of its weight, there is a $c\in\F\setminus\{0\}$,}
			&=c\cdot w(T_0)+(\text{terms of weight $\succ$ $w(T_0)$})
		\end{align*}
		from which it follows that $\det((HW(\vect))|_{S\times\bullet}\cdot M)$, and thus $\det((E(\vect)M)|_{S\times\bullet})$, are nonzero polynomials in $\vect$, as there is no cancellation of the nonzero term $c\cdot w(T_0)$.
	\end{innerproof}

	The above claim shows that the $n\times s$ matrix $E(\vect)M$ has an $s\times s$ minor of rank $s$, so $\rank_{\F(\vect)}(E(\vect)\cdot M)=s=\rank_\F(M)$ as desired.

	\uline{hitting set:} We derive the hitting set rank condenser from the generator via interpolation.  In particular, we need to form a large enough cube of evaluation points to ensure that the nonzero polynomial $\det((E(\vect)M)|_{S\times\bullet})$ assumes a nonzero value.  As the cube avoids zero in each entry, we can divide out the terms contributed by $\Lambda(\vect)$.  Thus, it remains to hit $\det((HW(\vect))|_{S\times\bullet}\cdot M)$.  Examining the above expression given by the Cauchy-Binet formula, it follows that $d_i^{\max}$ is the maximum degree in which $t_i$ appears, and $d_i^{\min}$ is the minimum. As $t_i\ne 0$, we can divide out by the minimum degree, leaving a polynomial with degree in $t_i$ of $d_i^{\max}-d_i^{\min}$, for each $i$.  Standard polynomial interpolation then shows that the determinant assumes a nonzero value at some point in the cube $C_1\times C_2\times\cdots$, so that $\rank_\F(E(\vect)M)=\rank_\F(M)$ at that point, as desired.
\end{proof}

\begin{remark}
	By taking $H$ so that $H_{i,j}=\omega^{ij}$, where $\omega$ has large multiplicative order, $H$ is then the parity check matrix of the dual Reed-Solomon code. The resulting rank condenser is the one of Forbes and Shpilka~\cite{ForbesShpilka12}, who used it (with some modification) to construct rank-metric codes meeting a bound of Roth~\cite{Roth91} over algebraically closed fields (Roth~\cite{Roth91} also gave such constructions).
\end{remark}

\subsection{Small-Support Rank Concentration via Rank Condensers}\label{sec:transfer matrix}

The above recipe shows how to construct a rank condenser in a certain way. Here, we show how to instantiate this recipe, by showing that the shifting map $\vecf(\vecx)\mapsto \vecf(\vecx+\vect)$ has code-like properties.  Consider how this shifting changes the coefficients of $\vecf$,
\begin{equation} \label{eq:transfer}
	\vecf(\vecx+\vect)
	=\sum_\vecb \coeff{\vecx^\vecb}{\vecf}(\vecx+\vect)^{\vecb}
	=\sum_\vecb \coeff{\vecx^\vecb}{\vecf}\sum_\veca{\vecb \choose \veca}\vecx^{\veca}\vect^{\vecb-\veca}.
\end{equation}
Thus, the coefficient of $\vecx^\veca$ in $\vecf(\vecx+\vect)$ equals $\sum_\vecb \coeff{\vecx^\vecb}{\vecf}{\vecb \choose \veca}\vect^{\vecb-\veca}$.  This linear transformation is summarized in the following construction.

\begin{construction}\label{construct-t}
	Let $n\ge 1$, $d\ge 0$.  Let $T(\vect)\in\F[\vect]^{\zr{d}^n\times \zr{d}^n}$ be the \emph{transfer} matrix of the shift map $\vecf(\vecx)\mapsto \vecf(\vecx+\vect)$ when acting on $n$-variate polynomials of individual degree $<d$, when these polynomials are expanded in the monomial basis.  That is, $T_{\veca,\vecb}=\binom{\vecb}{\veca} \vect^{\vecb-\veca}=\deriv_{\vecx^\veca}(\vecx^\vecb)(\vect)$.  
	
	Alternatively, this matrix changes basis, from expressing $\vecf$ in terms of its Hasse derivatives at $\vec{0}$ to expressing $\vecf$ in terms of its Hasse derivatives at $\vect$.

	Define $S_r$ to be the set of all monomials $\vecx^\veca$ of individual degree $<d$, such that $|\veca|_0\le \floor{\lg r}$. Define $T_r(\vect)$ to be the submatrix $T|_{S_r\times\bullet}$.
\end{construction}

\begin{remark*}
	For the reader acquainted with the work of Agrawal-Saha-Saxena~\cite{AgrawalSS12}, we note that our matrix $T$ is similar to the transfer matrix considered in that work. Specifically Equation $(4)$ and Lemma $10$ of \cite{AgrawalSS12} are captured by our definition of $T$ and Equation~\eqref{eq:transfer}. The main difference is that \cite{AgrawalSS12} state and prove their results in the language of Hadamard algebras while we simply work with polynomials. 
\end{remark*}

We will now work to show that $T_r$ is a rank condenser. To do so, we start with the following combinatorial lemma.  This lemma is similar to the sorting argument of Claim 18\ in Agrawal-Saha-Saxena~\cite{AgrawalSS12}, and is vaguely related to the notion of ``partial IDs'' as explored in Wigderson-Yehudayoff~\cite{WigdersonY12}.

\begin{lemma}\label{lem:partialid}
	Let $\vecb_1,\ldots,\vecb_r\in\Sigma^n$ be distinct strings.  Then there is an $i_0\in[r]$ and $|S|\subseteq[n]$ with $|S|\le\lg r$ such that $(\vecb_{i_0})|_S\ne (\vecb_{i})|_S$ for $i\ne i_0$.
\end{lemma}
\begin{proof}
	By induction on $r$.

	\uline{$r=1$:} This is vacuous, as $i\ne i_0$ never occurs, so $i_0=1$ and $S=\emptyset$ satisfy the requirements.

	\uline{$r>1$:} As the $\vecb_i$ are distinct, there is some coordinate $j\in[n]$ such that the $(\vecb_i)|_j$ are not all the same. Let $\sigma\in\Sigma$ be the minimum frequency symbol appearing as $(\vecb_i)|_j$ for some $i$.  Let $T\eqdef \{i:(\vecb_i)|_j=\sigma\}\subseteq[r]$, so that $|T|\le r/2$ by choice of $\sigma$.  Apply the induction hypothesis to the $\{\vecb_i:i\in T\}$ as strings in $\Sigma^{[n]\setminus\{j\}}$, which is possible as the strings in $T$ agree in the $j$-th position, so are still distinct with this position is removed. Thus yields an $i_0'\in T$ and a set $S'\subseteq[n]\setminus\{j\}$ such that $|S'|\le\lg r-1$ and $(\vecb_{i'_0})|_{S'}\ne (\vecb_{i})|_{S'}$ for $i\in T\setminus\{i'_0\}$. Take $i_0\eqdef i'_0$, and $S\eqdef S'\cup\{j\}$.  Then for $(\vecb_{i_0})|_{S}\ne (\vecb_{i})|_{S}$ for any $i\ne i_0$, where this disagreement occurs in $S'$ for $i\in T$, and occurs at position $j$ for $i\notin T$.
\end{proof}

This combinatorial lemma shows that in any set of strings there is a small ``fingerprint'' that distinguishes some string from all the others.  In our situation, strings will be (exponent vectors of) monomials, and we will try to distinguish monomials via derivatives.  The above fingerprint will allow our derivatives to be of small-support.  We now build these derivatives, starting with the univariate case.

\begin{lemma}\label{isolate-univar-deriv}
	Fix $d\ge 0$.  For $j\in\zr{d}$, there is a differential operator $\Delta_j=\sum_{k\in\zr{d}} c_{k} \deriv_{x^k}$ such that $\Delta_j(x^i)(1)=\ind{i=j}$.
\end{lemma}
\begin{proof}
	Consider the matrix $D\in\F^{\zr{d}\times\zr{d}}$ such that $D_{i,j}\eqdef\deriv_{x^i}(x^j)(1)=\binom{j}{i}$.  Note that this is an upper triangular matrix, with $1$'s along the diagonal.  Thus, this matrix is invertible, and let $C$ be its inverse. Define $c_{k}\eqdef C_{j,k}$. Thus $\Delta_j(x^i)(1)=\sum C_{j,k}\deriv_{x^k}(x^i)(1)=\sum C_{j,k} D_{k,i}=(C|_{j\times\bullet})\cdot(D|_{\bullet\times i})=\ind{i=j}$, as desired.
\end{proof}

Note that in the above, if we used normal (non-Hasse) derivatives then the diagonal of $D$ would not have $1$'s, and $D$ would not be invertible in sufficiently small characteristic.  We can now combine the fingerprint of strings, with the univariate derivatives, to yield multivariate derivatives isolating some monomial.

\begin{lemma}\label{isolate-multivar-deriv}
	Let $\vecx^{\vecb_1},\ldots,\vecx^{\vecb_r}$ be distinct monomials on the $n$ variables $\vecx$, and of individual degree $<d$.  Then there is a set $T\subseteq[n]$ with $|T|\le \floor{\lg r}$, a differential operator 
	\begin{equation}\label{ss-diff-op}
		\Delta=\sum_{\substack{\supp(\veca)\subseteq T\\|\veca|_\infty<d}} c_\veca \deriv_{\vecx^\veca}
	\end{equation}
	and an $i_0\in[r]$ such that $\Delta(\vecx^{\vecb_i})(1)=\ind{i=i_0}$.
\end{lemma}
\begin{proof}
	\autoref{lem:partialid} yields the set $T\subseteq[n]$ with $|T|\le
	\floor{\lg r}$, and the index $i_0\in[r]$.  We now construct $\Delta$.  For $j\in T$, let $\Delta_j=\sum_{k\in\zr{d}}c_{j,k} \deriv_{x_j^k}$ be the differential operator provided by \autoref{isolate-multivar-deriv} so that $\Delta_j(x_j^b)=\ind{b=(\vecb_{i_0})|_j}$. Now let $\Delta=\prod_{j\in T} \Delta_j$, where this product is the composition of the operators.  Recalling that $\deriv_{\vecx^\veca}\deriv_{\vecy^\vecb}=\deriv_{\vecx^\veca\vecy^\vecb}$ for disjoint variables $\vecx,\vecy$, we see that $\Delta$ has the form of \autoref{ss-diff-op}. Now consider $\Delta(\vecx^{\vecb_i})$.  As $\Delta_j$ only affects variable $x_j$, it follows that 
	\begin{align*}
		\Delta(\vecx^{\vecb_i})(\vec{1})
		&=\prod_{j\in T} \Delta_j \left(x_j^{(\vecb_i)|_j}\right)(1)\cdot\prod_{j\notin T} (x_j^{(\vecb_i)|_j})(1)\\
		&=\prod_{j\in T} \ind{(\vecb_i)|_j=(\vecb_{i_0})|_j}\cdot\prod_{j\notin T} 1\\
		&=\ind{i=i_0},
	\end{align*}
	where the last equality used the isolating properties of $T$.
\end{proof}

The above shows that we can distinguish some monomial from an entire set, using small-support derivatives.  Straightforward induction shows we can extend this to every monomial in the set.

\begin{corollarywp}\label{isolate-multivar-deriv-family}
	Let $\vecx^{\vecb_1},\ldots,\vecx^{\vecb_r}$ be distinct monomials on the $n$ variables $\vecx$, and of individual degree $<d$.  Then there is a permutation $\pi:[r]\to[r]$, and differential operators $\Delta_1,\ldots,\Delta_r$ such that
	\[
		\Delta_i(\vecx^{\vecb_{\pi(j)}})(\vec{1})=
		\begin{cases}
			1	&	i=j\\
			0	&	j<i
		\end{cases}.
	\]
	Further, each $\Delta_i$ is a linear combination of derivatives (of individual degree $<d$) with support contained in a set of size $\le \floor{\lg i}$.
\end{corollarywp}

The above corollary shows that for any set of monomials we can construct a set of differential operators whose application on these monomials induces a triangular matrix, which clearly has full rank.  We can thus use this to derive the needed properties of the $T_r$ matrix.

\begin{lemma}\label{deriv-are-code}
	Assume the setup of \autoref{construct-t}.  Then $T_r(\vec{1})$ is the parity check matrix of a distance $>r$ code.  That is, every $r$ columns are linearly independent.
\end{lemma}
\begin{proof}
	Let $\vecb_1,\ldots,\vecb_r$ be $r$ distinct monomials indexing columns in $T$. Consider the differential operators from \autoref{isolate-multivar-deriv-family}.  The operators $\Delta_i$ (evaluated at $\vec{1}$) acts on the monomials $\vecb_j$ and produces a vector $\vecv_i\in\F^r$.  By construction, each $\Delta_i$ is a linear combination of support $\le \floor{\lg r}$ derivatives, and thus $\vecv_i$ a linear combination of the rows of $T_r$ (when restricted to the columns induced by $\vecb_1,\ldots,\vecb_r$).  As the $\vecv_i$ form a triangular system (in some column ordering), it follows that the rows of $T_r$, when restricted to the columns induced by $\vecb_1,\ldots,\vecb_r$, have rank $r$.  As row-rank equals column-rank, it follows the columns indexed by the $\vecb_j$ are linearly independent, as desired.
\end{proof}

Thus, the $T_r(\vec{1})$ matrix is a parity check matrix of a code of distance $>r$. The above sections showed how to use such matrices to get a rank condenser.  Before we do so, we introduce the following definition.  It will allow us to encapsulate the essential properties of the shifts $\vect$ needed to ensure rank condensation.

\begin{definition}\label{defn-ind-deg-mon-map}
	Let $\vecg:\F^m\times\F^{m'}\to\F^n$ be a polynomial map.  It is an \emph{individual degree $<d$, $\ell$-wise independent monomial map} if for every $S\subseteq[n]$ of size $\le \ell$, there is an $\vecaa\in\F^{m'}$ such that polynomials $\{\vecg(\vect,\vecaa)^\veca\}_{\supp(\veca)\subseteq S,|\veca|_\infty<d}$ are nonzero, distinct monomials in $\vect$.
\end{definition}

As a trivial example, if we have $n$ independent variables, these form an $n$-wise independent monomial map.  Note the order of quantification, that the choice of $\vecaa$ can depend on the set $S$. When this is applied, the values for $\vecaa$ will be interpolated over some cube, which will implicitly union bound over all $S$. We now put the ingredients together to get our rank condenser.

\begin{corollary}\label{deriv-condense-rank}
	Assume the setup of \autoref{construct-t}.  Let $\vecg:\F^m\times\F^{m'}\to\F^n$ be an individual degree $<d$, $n$-wise independent monomial map.  Then $T_r(\vecg(\vect,\vecs))$ is a rank condenser for rank $r$.
\end{corollary}
\begin{proof}
	Fix the $\vecaa\in\F^{m'}$ guaranteed by the monomial map $\vecg$, when asking for independence on all $n$ variables, so that we can now regard $\vecg(\vect,\vecaa)$ as a polynomial map in $\F[\vect]$. Now note that $(T_r(\vect))_{\veca,\vecb}=\binom{\vecb}{\veca}\vect^{\vecb-\veca}=(\vect^\veca)^{-1}\cdot T_r(\vec{1})\cdot \vect^\vecb$. Thus, we can decompose $T_r(\vecg(\vect,\vecaa))=\Lambda(\vect,\vecaa)^{-1}\cdot T_r(\vec{1})\cdot W(\vect,\vecaa)$, where $\Lambda$, $T_r(\vec{1})$ and $W$ fit the hypothesis of \autoref{lem:rank-condenser-recipe}, as $\vecg$ is $n$-wise independent monomial map for degree $<d$, and by \autoref{deriv-are-code}. Thus, $T_r(\vecg(\vect,\vecaa))$ is a rank condenser for rank $r$.  Removing the substitution of $\vecs$, it follows that $T_r(\vecg(\vect,\vecs))$ is a rank condenser as well.
\end{proof}

Reinterpreting the above, we get the main result of this section, showing how to achieve rank concentration (recall \autoref{def:rank-concentration}).

\begin{corollary}\label{shift-rank-conc}
	Let $\vecf\in\F[\vecx]^r$ be polynomials of individual degree $<d$ on the $n$-variables $\vecx$. Let $\vecg:\F^m\times\F^{m'}\to\F^n$ be an individual degree $<d$, $n$-wise independent monomial map. Then $\vecf$ has support-$\floor{\lg r}$ rank concentration at $\vecg(\vect,\vecs)$ over the field $\F(\vect,\vecs)$.
\end{corollary}
\begin{proof}
	Expand the polynomials $\vecf$ in terms of their Hasse derivatives at zero, giving the matrix $M$, so that $M_{\vecb,i}=\deriv_{\vecx^\vecb}(f_i)(\vec{0})$.  Consider the multiplication $T_r(\vecg(\vect,\vecs)) \cdot M$.  By \autoref{deriv-condense-rank} it follows that $\rank_{\F(\vect,\vecs)}T_r(\vecg(\vect,\vecs))M=\rank_\F M$.  Thus, $T_r(\vecg(\vect,\vecs))M$ and $M$ have the same $\F(\vect,\vecs)$-row-span.  From the definition of $T_r$ and \autoref{eq:transfer}, we then get that
	\begin{align*}
		\Span_{\F(\vect,\vecs)}\{\deriv_{\vecx^\veca}(\vecf)(\vecg(\vect,\vecs))\}_{\substack{|\veca|_0\le\floor{\lg r}\\|\veca|_\infty<d}}
		&=\rspan_{\F(\vect,\vecs)}(T_r(\vecg(\vect,\vecs))M)\\
		&=\rspan_{\F(\vect,\vecs)}(M)\\
		&=\Span_{\F(\vect,\vecs)}\{\deriv_{\vecx^\veca}(\vecf)(\vec{0})\}_\veca\\
		&=\Span_{\F(\vect,\vecs)}\{\deriv_{\vecx^\veca}(\vecf)(\vecg(\vect,\vecs))\}_\veca
	\end{align*}
	where we used \autoref{change-deriv-point} in the last equality to change the point of derivation.
\end{proof}

\section{Hitting Sets for Commutative ROABPs}\label{sec:commutative}

In this section, we give $n^{\O(\lg n)}$-size hitting sets for polynomials that are computed by ROABPs in \emph{all} variable orders, and thus are in some sense \emph{commutative}.  This class is also characterized by having low \emph{evaluation dimension} a concept of Saptharishi, as reported in Forbes-Shpilka~\cite{ForbesShpilka12a}, who gave $n^{\O(\lg n)}$-size hitting sets for this class.  The analysis of the hitting sets will follow quickly from the rank concentration results of \autoref{sec:transfer matrix}, and can be seen as an alternative proof of some of the results of the work of Agrawal-Saha-Saxena~\cite{AgrawalSS12}.\footnote{This section does not give $n^{\O(\lg n)}$-size hitting sets for depth-3 set-multilinear polynomials, one of the core results of Agrawal-Saha-Saxena~\cite{AgrawalSS12}. The results of \autoref{sec:noncommutative} do cover this class of polynomials, by giving a $n^{\O(\lg^2 n)}$-size hitting set for a more general model. However, the techniques of this section do recover the Agrawal-Saha-Saxena~\cite{AgrawalSS12} result, essentially by changing $x_i^j\mapsto x_{i,j}$ in all of our arguments.  This changes the transfer matrix, and the set of monomials that the monomial map must preserve.  However, the techniques (to show the transfer matrix is a rank condenser, that we can shift using low-wise independent maps, etc.) can be adapted to these changes in a straightforward manner.}  While this does not improve upon the parameters of the results of Agrawal-Saha-Saxena~\cite{AgrawalSS12} or Forbes-Shpilka~\cite{ForbesShpilka12a}, we will see in \autoref{sec:hashing+FS} how to combine the rank concentration ideas with the generator of Forbes-Shpilka~\cite{ForbesShpilka12a} (using techniques from hashing, and the generator of Shpilka-Volkovich~\cite{ShpilkaVolkovich09}) to obtain $n^{\O(\lg\lg n)}$ size hitting sets for certain subclasses of these polynomials, such as diagonal circuits.

We begin by showing that for commutative ROABPs, establishing small-support rank concentration for all $n$ variables follows from small-support rank concentration for all size-$\Omega(\lg r)$ subsets of variables.  Establishing rank concentration for these small subsets of variables can then be done by brute-force using the rank condensers of \autoref{sec:transfer matrix}.

\begin{theorem}\label{thm:rank-concentration-comm}
	Let $F(\vecx)\in\F[\vecx]^{r\times r}$ be of individual degree $<d$, and be computed by a width-$r$ commutative ROABP.  Let $\vecg(\vect,\vecs)$ be an individual degree $<d$, $(\floor{\lg r^2}+1)$-wise independent monomial map.  Then, $F(\vecx)$ has support-$\floor{\lg r^2}$ rank concentration at $\vecg(\vect,\vecs)$ over the field $\F(\vect,\vecs)$.
\end{theorem}
\begin{proof}
	Consider some derivative $\deriv_{\vecx^{\veca_0}}(F)$, with ${\veca_0}\in\zr{d}^n$.  We wish to show that this derivative, evaluated at $\vecg(\vect,\vecs)$, is contained in the $\F(\vect,\vecs)$-span of the derivatives of small-support, evaluated at $\vecg(\vect,\vecs)$.  Thus, if the derivative is already support $\le \floor{\lg r^2}$, then this is trivial.  Now consider when this is not true.  Thus, we can partition the monomial $\vecx^{\veca_0}$ into $\vecy^{\vecb_0}\vecz^{\vecc_0}$ such that $\vecy$ has $\floor{\lg r^2}+1$ variables and ${\vecb_0}$ has full-support. 

	As $F$ is computable by a width-$r$ commutative ROABP, we can express it then as $F(\vecx)=F(\vecy,\vecz)=G(\vecy) \cdot H(\vecz)$, for $G\in\F[\vecy]^{r\times r}$ and $H\in\F[\vecz]^{r\times r}$.\footnote{This is where we use the commutativity of the ROABP.}  By bilinearity of matrix multiplication, we can factor the derivative as $\deriv_{\vecx^{\veca_0}}(F)=\deriv_{\vecy^{\vecb_0}}(G) \cdot \deriv_{\vecz^{\vecc_0}}(H)$.  Now observe that the output of $\vecg(\vect,\vecs)$, when restricted to $\vecy$, is still an individual degree $<d$, $(\floor{\lg r^2}+1)$-wise independent monomial map.  Thus, as $\vecy$ has $\le \floor{\lg r^2}+1$ variables, it follows from \autoref{shift-rank-conc} that $G(\vecy)$ has support-$\floor{\lg r^2}$ rank concentration at $\vecg(\vect,\vecs)|_\vecy$, over $\F(\vect,\vecs)$.  Rephrasing, this means that
	\[
		\deriv_{\vecy^{\vecb_0}}(G)(\vecg(\vect,\vecs)|_\vecy)\in\Span_{\F(\vect,\vecs)}\{\deriv_{\vecy^\vecb}(G)(\vecg(\vect,\vecs)|_\vecy)\}_{|\vecb|_0\le\floor{\lg r^2}}.
	\]
	Invoking again the bilinearity of matrix multiplication, it follows that we can multiply the above equation by $\deriv_{\vecz^{\vecc_0}}(H)$, yielding
	\begin{align*}
		\deriv_{\vecx^{\veca_0}}(F)(\vecg(\vect,\vecs))
		&=\deriv_{\vecy^{\vecb_0}}(G)(\vecg(\vect,\vecs)|_\vecy)\cdot \deriv_{\vecz^{\vecc_0}}(H)(\vecg(\vect,\vecs)|_\vecz)\\
		&\in\Span_{\F(\vect,\vecs)}\{\deriv_{\vecy^\vecb}(G)(\vecg(\vect,\vecs)) \deriv_{\vecz^{\vecc_0}}(H)(\vecg(\vect,\vecs)|_\vecz)\}_{|\vecb|_0\le\floor{\lg r^2}}\\
		&\subseteq\Span_{\F(\vect,\vecs)}\{\deriv_{\vecx^\veca}(F)(\vecg(\vect,\vecs))\}_{|\veca|_0<|\supp(\veca_0)|}.
	\end{align*}

	Thus, this shows that for any derivative $\deriv_{\vecx^{\veca_0}}(F)$ with support $>\floor{\lg r^2}$, its evaluation at $\vecg(\vect,\vecs)$ is contained in the $\F(\vect,\vecs)$-span of derivatives with smaller support.  Applying this argument repeatedly (each time with a new partition $\vecx^{\veca_0}=\vecy^{\vecb_0}\vecz^{\vecc_0}$) thus yields the claim.
\end{proof}

Thus, the above shows that given the appropriate monomial map, we can achieve small-support rank concentration for commutative ROABPs.  We now observe that the generator of Shpilka and Volkovich~\cite{ShpilkaVolkovich09} gives such a monomial map. 

\ignore{Note that the usage of the SV-generator in hitting small support monomials, as given in \autoref{lem:GSV-hit-ss}, we could have replaced the SV-generator with a hitting set consisting of all low-hamming-weight strings.  This replacement is not possible with a monomial map, as we are crucially using generators as opposed to hitting sets, as was also done in Shpilka-Volkovich~\cite{ShpilkaVolkovich09}.}

\begin{lemma}\label{sv-mon-map}
	Let $\F$ be a field with $|\F|>n$. Then the SV-generator $\GSV{n,\ell}:\F^\ell\times\F^\ell\to\F^n$ of \autoref{thm:SV generator} is an individual degree $<d$, $\ell$-wise independent monomial map.
\end{lemma}
\begin{proof}
	As argued in \autoref{lem:GSV-hit-ss}, for any subset $S\subseteq[n]$ of size $\le \ell$, we can set the $\vecz$ variables appropriately such that $|S|$ of the $\vecy$ variables are exactly planted into the positions $\vecx|_S$.  It is then clear that this is a monomial map, and all individual degree $<d$ monomials of this map are distinct.
\end{proof}

Thus, we have seen that shifting by a monomial map yields rank concentration, how to construct such a monomial map, and how rank concentration yields hitting sets.  Putting this all together, we obtain hitting sets as presented in the following corollary.

\begin{corollary}\label{hit-set-commut}
	Let $|\F|>nd$. Let $f(\vecx)\in\F[\vecx]$ be of individual degree $<d$, and be computed by a width-$r$ commutative ROABP.  Then $f\not\equiv 0$ iff $f\circ \GSV{n,1+2\floor{2\lg r}}\not\equiv 0$. 

	In particular, there is a $\poly(n,d)^{\O(\lg r)}$ size $\poly(n,d,r)$-explicit hitting set for width-$r$ commutative ROABPs on $n$ variables, of individual degree $<d$.
\end{corollary}
\begin{proof}
	\uline{generator:} We have by definition that $f$ is a linear combination of the entries of $F(\vecx)$, where $F(\vecx)$ is computed by a matrix-valued commutative ROABP (recall \autoref{lem:roabp-as-imm}). By \autoref{thm:rank-concentration-comm} it follows that $F\circ \GSV{n,\floor{2\lg r}+1}$ is support-$\floor{\lg r^2}$ rank concentrated.  Thus, \autoref{cor:rank-conc-to-hitting-sets} shows that $f\not\equiv0$ iff $f\circ \left(\GSV{n,\floor{2\lg r}}+\GSV{n,\floor{2\lg r}+1}\right)\not\equiv 0$, where $\GSV{n,\floor{2\lg r}}(\vecy,\vecz)$ and $\GSV{n,\floor{2\lg r}+1}(\vecy',\vecz')$ are on disjoint variables.  However, the SV-generator construction of \autoref{thm:rank-concentration-comm} is additive, in that as polynomial maps we have the equality 	$\GSV{n,\floor{2\lg r}}(\vecy,\vecz)+\GSV{n,\floor{2\lg r}+1}(\vecy',\vecz')=\GSV{n,1+2\floor{2\lg r}}(\vecy'',\vecz'')$, where $\vecy''=(\vecy,\vecy')$ and $\vecz''=(\vecz,\vecz')$.

	\uline{hitting set:} The hitting set is the evaluation of $\GSV{n,1+2\floor{2\lg r}}$ on a sufficiently large product space of inputs, and the correctness follows from interpolation. That is, $f\circ\GSV{}$ only has $1+2\floor{2\lg r}$ variables, and  individual degree $<nd$.
\end{proof}

Note that the SV-generator is in some sense an algebraic analogue of $\ell$-wise independence.  However, it is not the case that \emph{any} $\ell$-wise independent map will fool commutative ROABPs. For example, the polynomial $x_1+\cdots+x_n$ is computable in this model, but is not fooled by the $n-1$ independent polynomial map given by $(x_1,x_2,\ldots,x_{n-1},-(x_1+\cdots+x_{n-1}))$.

To improve the above hitting sets of size $n^{o(\lg n)}$, note that there are two sources of $n^{\Theta(\lg n)}$ factors in the above construction: the monomial map used to shift into rank concentration, and the conversion from rank concentration to hitting sets.  In \autoref{sec:hashing+FS}, we will see an improved way to convert from rank concentration to hitting sets, that will only result in an $n^{\O(\lg\lg n)}$ overheard, but will be specific for commutative ROABPs.  Unfortunately, no such improvement to the monomial map is yet known, but one can give a better monomial map when the individual degree of the polynomials is small.  This alternative monomial map is based on the hashing ideas of Klivans-Spielman~\cite{KlivansSpielman01}.

\begin{lemma}\label{mon-map-KS}
	Let $|\F|\ge\poly(n,d)$.  Then the KS-generator $\GKS{n,d,d^\ell}:(\F^{\O(\ell\log_{nd} d)})^2\to\F^n$ of \autoref{thm:KS generator} is a $\poly(n,d,\ell)$-explicit, individual degree $<d$, $\ell$-wise independent monomial map.
\end{lemma}
\begin{proof}
	Note that for any subset $S\subseteq[n]$ of size $\le\ell$, there are at most $d^\ell$ monomials of individual degree $<d$.  Thus, it follows from \autoref{thm:KS generator} that the KS-generator with sparsity $s=d^\ell$ will map these monomials distinctly, as desired.
\end{proof}

Note that the above KS-generator monomial map is better than the SV-generator monomial map when $d$ is small, as we will be taking $\ell=\lg r$.  As the individual degree of the KS-generator is $\poly(n,d)$, and we have $\O(\log_{nd} d^\ell)$ variables, interpolating through this generator yields $\poly(n,d)^{\O(\ell \log_{nd} d)}=\poly(d^\ell)$ many points.  In particular, when $d=\O(\polylog(r,n))$, this will yield $\poly(n,r)^{\O(\lg\lg (nr))}$ shifts such that one shift will yield rank concentration. In contrast, the SV-generator adds $\ell$ new variables, and multiplies their degree from $d$ to $dn$, as the generator has degree $n$.  Thus, even when $d$ is small the number of shifts needed will be $\Omega(n^{\lg r})$.  We will exploit this further in \autoref{sec:hashing+FS}, where we show techniques for reducing the complexity of hitting polynomials with rank concentration.  

Lastly, we note that while the SV-generator is sufficient for the commutative ROABP results, it will not be sufficient for the results on general ROABPs with unknown orders in \autoref{sec:noncommutative}.  Instead, the KS-generator will be used.

\section{Condensing Rank from Matrices with a Known Basis}\label{sec:condense-known-basis}

This section will develop a new recipe for constructing (a variant of) rank condensers from error-correcting codes, and will then instantiate it with the transfer matrix of \autoref{sec:transfer matrix}.  That is, while the rank condenser recipe of \autoref{lem:rank-condenser-recipe} is sufficient for the results in \autoref{sec:commutative} on commutative ROABPs, the results of \autoref{sec:noncommutative} about unknown-order ROABPs need rank condensers that have smaller ``seed length'' than the above result.  In the context of \autoref{lem:rank-condenser-recipe}, this means that the condition on the matrix $W$, that all monomials are distinct, is too stringent.  We will want to apply a rank condenser in the setting with many variables, and keeping all such monomials distinct would prevent the construction of generators inputting few variables. To achieve better parameters, we will weaken the requirements on the rank condenser.  

In particular, observe that the above rank condenser is in some sense ``information theoretic'', as the matrices involved are all completely written down, without any notion of ``simplicity''.  To achieve better parameters, we can impose that the matrices we wish to extract from are simple in some sense.  Here, we will restrict to condensing rank from matrices that have a row-basis in some known ``medium-sized'' subset of the rows --- small enough so the seed length can decrease, but large enough so that the problem is non-trivial. In particular, we will try to construct the following type of rank condenser.

\begin{definition}[Rank Condenser Generator for Known Bases]
	A matrix $E(\vect) \in \F[\vect]^{n\times m}$ is said to be a \emph{(seeded) rank condenser (generator) for rank $r$ and rows $P\subseteq[m]$} if for every matrix $M \in \F^{m\times k}$ of rank $\le r$ such that $\rank_\F(M)=\rank_\F(M|_{P\times \bullet})$, we have that $\rank_{\F(\vect)}(E(\vect)\cdot M)=\rank_\F(M)$. 
\end{definition}

An obvious approach to construct the above would be to ignore the rows outside $P$, by zeroing out parts of the rank condenser, and plugging in a rank condenser for the $P$.  However, this is not compatible with the black-box access of polynomials.  Thus, we must consider rank condensers that naturally arise in the shifting/partial-derivative process.

In particular, we will relax the restriction, on the matrix $W$ of \autoref{lem:rank-condenser-recipe}, that all monomials are distinct. Recall though, that in the above rank condenser, this distinctness property drove the uniqueness of the minimum weight term, which was crucial in showing that some determinant does not vanish in the Cauchy-Binet sum. If non-distinctness of these weights is allowed, then the weight minimization may not be unique and hence the corresponding monomial could be canceled out. To ensure uniqueness, it will be enough that the \emph{small weights} are distinct. In the rank condenser context, this will correspond to the rows in the $P$ getting distinct weights, and that these distinct weights are smaller than any non-distinct weights. As this regime in greedy/matroid optimization seems somewhat less standard, we prove the required uniqueness property.  For concreteness, we only discuss this optimization within our context, but the results apply to any matroid.  We begin with a property of the greedy algorithm.

\begin{lemma}\label{lem:matroidgreedy}
	Let $M\in\F^{n\times m}$ be a matrix of rank $r$.  Using some total order $\prec$ on $[m]$, greedily choose $r$ linearly independent columns (at each step, pick the least-$\prec$-indexed column that increases the rank of the chosen vectors), and denote this set $T=\{i_1,\ldots,i_r\}$ with $i_1\prec\cdots\prec i_r$. Let $T'\subseteq[m]$ be any set of $r$ linearly independent columns, and denote $T'=\{i'_1,\ldots,i'_r\}$ with $i'_1\prec\cdots\prec i'_r$. Then $T\preceq  T'$ point-wise, that is, $i_j\preceq i'_j$ for $j\in[r]$.  Further, if $T\ne T'$, then $i_{j_0}\prec i'_{j_0}$ for some $j_0\in[r]$.
\end{lemma}
\begin{proof}
	Suppose not, so that $T\not\preceq  T'$, for contradiction. Let $j_0$ be the first $j\in[r]$ such that $i_j\succ i'_j$. Denote $T_k=\{i_1,\ldots,i_k\}$ and $T'_k=\{i'_1,\ldots,i'_k\}$ for $k\in[r]$.  Thus $T_{j_0-1}$ indexes $r-1$ linearly independent columns and $T'_{j_0}$ indexes $r$ linearly independent columns.  By the Steinitz Exchange Lemma of linear algebra, it follows there is some $i'_{j_1}\in T'_{j_0}\setminus T_{j_0-1}$ such that $T_{j_0-1}\cup\{i'_{j_1}\}$ indexes $r$ linearly independent columns.  Using that $i'_{j_1}\preceq i'_{j_0}\prec i_{j_0}$, it follows that there is some $j_2\in[j_0-1]$ such that $i_{j_2}\prec i'_{j_1}\prec i_{j_2+1}$.
	
	When $T_{j_2}$ is being greedily extended to $T_{j_2+1}$, the least rank-increasing column after all elements in $T_{j_2}$ will be chosen.  By hypothesis, $i_{j_2+1}$ was chosen.  However, as adding $i'_{j_1}$ increases the rank of $T_{j_0-1}$, and $T_{j_2}\subseteq T_{j_0-1}$, it follows that adding $i'_{j_1}$ increases the rank of $T_{j_2}$.  As $i_{j_2}\prec i'_{j_1}\prec i_{j_2+1}$, it follows that $i_{j_2+1}$ could not have been chosen, as $i'_{j_1}$ is a better choice, yielding the desired contradiction.  Thus $T\preceq  T'$ point-wise, as desired.

	If $T\ne T'$, then not all indices can be equal, so by the above $i_{j_0}\prec i'_{j_0}$ for some $j_0$.
\end{proof}

We next state the uniqueness lemma when minimizing the monomial weight of bases of the columns of a matrix, when only the small weights are distinct.  As above, this statement can be interpreted for general greedy/matroid optimization, but for concreteness we only discuss our context.  Note that the usual setting for greedy/matroid optimization is maximization with distinct weights.  In the below setting, we have minimization and thus we minimize only over full-rank sets of columns, as else the empty set is trivially the minimum weight set.  Further, in our applications it is crucial that we allow some weights to be repeated.  While the resulting uniqueness result follows from standard arguments, it is somewhat non-standard so we include it for completeness.

\begin{lemma}\label{lem:matroidgreedypartial}
	Let $M\in\F^{n\times m}$ be a matrix of rank $r$.  Let $\prec$ be a monomial ordering on $\F[\vecx]$, and let $w:[m]\to \F[\vecx]$ weight the columns in $M$ with monomials in $\F[\vecx]$.  Suppose the partition $[m]=I\sqcup J$ is so that $w$ is injective on $I$, and that $w(i)\prec w(j)$ for all $i\in I$ and $j\in J$.  Suppose further that the $\rank(M|_{\bullet\times I})=r$.  Then there is a unique set $T\subseteq[m]$ of size $r$, with $\rank(M|_{\bullet\times T})=r$, that minimizes $w(T)\eqdef \prod_{i\in T} w(i)$ with respect to the monomial order.  Further, $T\subseteq I$.
\end{lemma}
\begin{proof}
	The monomial order, as a total order on $\F[\vecx]$, induces a partial order on $[m]$ via the weight function $w$.  Extend this partial order arbitrarily to a total order on $[m]$, and denote it by `$\prec$' also, abusing notation.  From this it follows that for $i,i'\in[m]$, $i\prec i'$ implies $w(i)\preceq w(i')$.  Thus, it follows that $I$ comes before $J$ in the $\prec$-order on $[m]$.
	
	Choose the set $T$ of $r$ linearly independent columns greedily, using the $\prec$-order on $[m]$.  As there is some set of $r$ linearly independent columns $T'\subseteq I$, it follows from \autoref{lem:matroidgreedy} that $T\preceq T'$ point-wise, so then $T\subseteq I$, as $I$ comes before $J$ in the $\prec$-order on $[m]$. Now consider any set $T'\subseteq[m]$ of $r$ linearly independent columns, with $T\ne T'$.  Again, $T\preceq T'$, so in particular $w(T)\preceq w(T')$.  However, using \autoref{lem:matroidgreedy} further, we see there is some index $j_0\in[r]$ so that $i_{j_0}\prec i'_{j_0}$ when we write $T=\{i_1,\ldots,i_r\}$ and $T'=\{i'_1,\ldots,i'_r\}$.  As $i_{j_0}\in T\subseteq I$, and the weights in $I$ are distinct and all strictly below the weights in $J$, it follow that $w(i_{j_0})\prec w(i'_{j_0})$.  From this it follows that $w(T)\prec w(T')$.  Thus $T\subseteq I$ is the unique minimizer of the weight $w$ over sets of $r$ linearly independent columns.
\end{proof}

To use this uniqueness lemma in our rank condenser, we need a notion of what a ``small'' weight will be.  Here, we will use a new variable $u$, such that the smallness of a monomial will correspond to its degree in $u$.

\begin{lemma}\label{lem:rank-condenser-known-basis}
	Fix $d\ge 0$. Let $E_d(\vect,u)\in\F[\vect,u]^{n\times m}$, with $E_d=\Lambda(\vect,u)^{-1}\cdot H\cdot W(\vect,u)$ such that
	\begin{itemize}
		\item $\Lambda(\vect,u)\in\F[\vect,u]^{n\times n}$ is a diagonal matrix, whose diagonal entries are nonzero monomials.
		\item $H\in\F^{n\times m}$ is the parity check matrix of an error correcting code of distance $>r$.  That is, every $r$ columns of $H$ are linearly independent.
		\item $W(\vect,u)\in\F[\vect,u]^{m\times m}$ is a diagonal matrix, whose diagonal entries are nonzero monomials.  The monomials in $W(\vect,u)|_{P_d\times P_d}$ are distinct, where $P_d=\{j:j\in[m], \deg_u W(\vect,u)_{j,j}< d\}$.
	\end{itemize}
	Then $E(\vect,u)$ is a rank condenser for rank $r$ and rows $P_d$.
\end{lemma} 
\begin{proof}
	Consider any matrix $M\in\F^{m\times k}$ of rank $s\le r$ with $\rank_\F(M)=\rank_\F(M|_{P_d\times\bullet})$.

	\uline{$k>s$:} Choose the submatrix $M'\in\F^{m\times s}$ of $M$, such that $M'$ has rank $s$.  As $\rank_\F(M|_{P_d\times\bullet})=s$ we can choose $M'$ so that $\rank_\F(M'|_{P_d\times\bullet})=s$ also.  Now apply the $k=s$ case, just as in the proof of \autoref{lem:rank-condenser-recipe}.

	\uline{$k=s$:} As in the proof of \autoref{lem:rank-condenser-recipe}, $W$ induces a weight function $w:[m]\to\F[\vect,u]$. Let $\prec$ be a lexicographic monomial order on $\F[\vect,u]$, such that $u$ comes before all variables in $\vect$.  It then follows from that \autoref{lem:matroidgreedypartial} applies, using the partition $[m]=P_d\sqcup ([m]\setminus P_d)$, so that there is a unique set $T_0\subset [m]$ such that $\det(M|_{T_0\times\bullet})\ne 0$ minimizing $w(T_0)$ with respect to the $\prec$-order.  The rest of the proof then follows exactly as in the proof of \autoref{lem:rank-condenser-recipe}.
\end{proof}

With this new recipe, we now instantiate it with the transfer matrix of \autoref{sec:transfer matrix}.  To do so, we need now the notion of a monomial map with independence with respect to \emph{total degree}.

\begin{definition}\label{defn-tot-deg-mon-map}
	Let $\vecg:\F^m\times\F^{m'}\to\F^n$ be a polynomial map.  It is a \emph{total-degree-$(<D)$ independent monomial map} if there is an $\vecaa\in\F^{m'}$ such that the polynomials $\{\vecg(\vect,\vecaa)^\veca\}_{|\veca|_1<D}\subseteq\F[\vect]$ are nonzero, distinct monomials in $\vect$.
\end{definition}

Note that the main difference from \autoref{defn-ind-deg-mon-map} is that now we consider monomials in many variables whereas previously we only considered monomials in $\ell$ variables.  We also now restrict the total degree, so in the above $D$ should be considered small. The variable $u$ of \autoref{lem:rank-condenser-known-basis} is used to grade with respect to degree, and thus take advantage the the total degree is not too large. 

We now use such maps to get a rank condenser.

\begin{corollary}\label{deriv-condense-rank-known-basis}
	Assume the setup of \autoref{construct-t}.  Let $\vecg:\F^m\times\F^k\to\F^n$ be an total-degree-$(<D)$ independent monomial map.  Then $T_r(u\cdot \vecg(\vect,\vecs))$ is a rank condenser for rank $r$ and rows corresponding to monomials of total degree $<D$.
\end{corollary}
\begin{proof}
	Fix the $\vecaa\in\F^{m'}$ guaranteed by the monomial map $\vecg$, so that we can now regard $\vecg(\vect,\vecaa)$ as a polynomial map in $\F[\vect]$.  As in \autoref{deriv-condense-rank}, we can express $T_r(u\vecg(\vect,\vecaa))=\Lambda(u,\vect,\vecaa)^{-1}\cdot T_r(\vec{1})\cdot W(u,\vect,\vecaa)$.  As before, $\Lambda$ and $T_r(\vec{1})$ fit the hypothesis of \autoref{lem:rank-condenser-known-basis}, where for the latter we use \autoref{deriv-are-code}.  Further, $W(u,\vect,\vecaa)$ has the desired monomial distinctness properties as $\vecg$ is a monomial map with independence up to total degree $<D$. Thus, $T_r(u\vecg(\vect,\vecaa))$ is a rank condenser for rank $r$.  Removing the substitution of $\vecs$, it follows that $T_r(u\vecg(\vect,\vecs))$ is a rank condenser as well.
\end{proof}

Reinterpreting the above in terms of rank concentration, we get the main result of this section. 

\begin{corollary}\label{shift-rank-conc-known-basis}
	Let $\vecf\in\F[\vecx]^r$ be polynomials of individual degree $<d$ on the $n$-variables $\vecx$. Let $\vecg:\F^m\times\F^{m'}\to\F^n$ be an total degree-$(<d\ell)$ independent monomial map. Then if $\vecf$ has support-$\ell$ rank concentration at $\vecaa\in\F^n$, then $\vecf$ has support-$\floor{\lg r}$ rank concentration at $\vecaa+u\cdot\vecg(\vect,\vecs)$ over the field $\F(u,\vect,\vecs)$.
\end{corollary}
\begin{proof}
	Expand the polynomials $\vecf$ in terms of their Hasse derivatives at $\vecaa$, giving the matrix $M$, so that $M_{\vecb,i}=\deriv_{\vecx^\vecb}(f_i)(\vecaa)$.  Let $P_{d\ell}\subseteq \zr{d}^n$ index the rows of $M$ given by monomials of total degree $<d\ell$.  Thus $P_{d\ell}$ contains all monomials with support $\le \ell$.  Thus, by hypothesis $\rank_\F(M)=\rank_\F(M|_{P_{d\ell}\times\bullet})$.
	
	\sloppy Now consider the multiplication $T_r(u\vecg(\vect,\vecs)) \cdot M$.  By \autoref{deriv-condense-rank-known-basis} it follows that $\rank_{\F(u,\vect,\vecs)}T_r(u\vecg(\vect,\vecs))M=\rank_\F M$.  Thus, $T_r(u\vecg(\vect,\vecs))M$ and $M$ have the same $\F(u,\vect,\vecs)$-row-span.  Thus, as in \autoref{shift-rank-conc}, we have that the definition of $T_r$ and \autoref{eq:transfer} imply that,
	\begin{align*}
		\Span_{\F(u,\vect,\vecs)}\{\deriv_{\vecx^\veca}(\vecf)(\vecaa+u\vecg(\vect,\vecs))\}_{\substack{|\veca|_0\le\floor{\lg r}\\|\veca|_\infty<d}}
		&=\rspan_{\F(u,\vect,\vecs)}(T_r(u\vecg(\vect,\vecs))M)\\
		&=\rspan_{\F(u,\vect,\vecs)}(M)\\
		&=\Span_{\F(u,\vect,\vecs)}\{\deriv_{\vecx^\veca}(\vecf)(\vecaa)\}_\veca\\
		&=\Span_{\F(u,\vect,\vecs)}\{\deriv_{\vecx^\veca}(\vecf)(\vecaa+u\vecg(\vect,\vecs))\}_\veca
	\end{align*}
	where, again, we used \autoref{change-deriv-point} in the last equality to change the point of derivation.
\end{proof}

Note that $\vecaa$ could be a polynomial map in auxiliary variables $\vect'$ (and this will happen in later sections), in which case the above lemma would use the field $\F(\vect')$.

\section{Hitting Sets for Unknown Order ROABPs}\label{sec:noncommutative}

In this section, we present hitting sets of size $\poly(d,n,r)^{d\lg r\cdot \lg n}$ for ROABPs on $n$ variables, width-$r$ and individual degree $<d$, when the order of the variables is unknown.  While this gives worse parameters than the size-$\poly(d,n,r)^{\lg n}$ hitting set of Forbes-Shpilka~\cite{ForbesShpilka12a} results for ROABPs, that hitting set only worked for a fixed variable order.  Further, while the results of this section have a poor dependence on the individual degree $d$, note that $d=2$ is sufficiently interesting, as it covers all of the set-multilinear models considered by Agrawal-Saha-Saxena~\cite{AgrawalSS12}.

The results of \autoref{sec:commutative} showed that, in commutative ROABPs, rank concentration can be achieved in $n$ variables by reducing it to achieving rank concentration in all size-$\Theta(\lg r)$ subsets of variables.  Then, rank concentration can be brute-force achieved in quasipolynomial time.  However, this connection crucially used commutativity to reorder the variables, and seems to break down otherwise.  

In considering general ROABPs, we will still use shifts to concentrate rank on smaller-support derivatives.  However, instead of relying on a single shift, we will shift $\lg n$ times, each time halving the support of the derivatives needed.  While the rank concentration results of \autoref{sec:transfer matrix} can achieve this, such results are too expensive when applied to all $n$ variables at once.  Instead, we will use the rank concentration results of \autoref{sec:condense-known-basis}, which allow each shift to build on previous shifts.

More specifically, suppose we have two ROABPs $F(\vecx)$ and $G(\vecy)$, each with support-$\ell$ rank concentration.  First, we observe a simple merge, that is, $F(\vecx)G(\vecy)$ has support-$2\ell$ rank concentration.  By itself, this merge step contains no work.  Now we seek to reduce the support. This can be done with a rank condenser, and in particular, since $FG$ has support-$2\ell$ rank concentration, $FG$ has a ``known basis'' that we can condense.  In this case, the result of \autoref{shift-rank-conc-known-basis} will condense to support-$\lg r^2$, using a single shift of $(\vecx,\vecy)$.  Thus, if we take $\ell=\lg r^2$, then this merge-and-reduce process has doubled the number of variables, while maintaining support-$\ell$ rank concentration.  By applying this $\lg n$ times, we can shift $n$-variate ROABP into small-support rank concentration.

Given the proof outline, we begin the formalities by observing that rank concentration of two ROABPs on disjoint variables implies rank concentration of their product.  The proof shares the bilinearity-based ideas of \autoref{thm:rank-concentration-comm}, but we must now preserve the order of the variables.

\begin{lemma}\label{roabp-combine}
	Let $\vecx$ and $\vecy$ be disjoint sets of $n$ variables.
	Consider $F(\vecx)\in\F[\vecx]^{r\times r}$ and $G(\vecy)\in\F[\vecy]^{r\times r}$.  Then if $F(\vecx)$ is support-$\ell$ rank concentrated at $\vecaa\in\F^n$ and $G(\vecy)$ is support-$k$ rank concentrated at $\vecbb\in\F^n$, then $F(\vecx)G(\vecy)$ is support-$(\ell+k)$ rank concentrated at $(\vecaa,\vecbb)$.
\end{lemma}
\begin{proof}
	Consider a derivative $\deriv_{\vecx^{\veca_0}\vecy^{\vecb_0}}(FG)$.  The rank concentration implies that
	\[
		\deriv_{\vecx^{\veca_0}}(F)(\vecaa)
		\in\Span_\F \{\deriv_{\vecx^\veca}(F)(\vecaa)\}_{|\veca|_0\le \ell}
	\]
	and
	\[
		\deriv_{\vecy^{\vecb_0}}(G)(\vecbb)
		\in\Span_\F \{\deriv_{\vecy^\vecb}(G)(\vecbb)\}_{|\vecb|_0\le k}
	\]
	Thus, as matrix multiplication is bilinear, and the variables $\vecx$ and $\vecy$ are disjoint, it follows that
	\begin{align*}
		\deriv_{\vecx^{\veca_0}\vecy^{\vecb_0}}(FG)
		&=\deriv_{\vecx^{\veca_0}}(F)(\vecaa)		\deriv_{\vecy^{\vecb_0}}(G)(\vecbb)\\
		&\in\Span_\F \{\deriv_{\vecx^\veca}(F)(\vecaa)\deriv_{\vecy^\vecb}(G)(\vecbb)\}_{|\veca|_0\le \ell,|\vecb|_0\le k}\\
		&=\Span_\F \{\deriv_{\vecx^\veca\vecy^\vecb}(FG)(\vecaa,\vecbb)\}_{|\veca|_0\le \ell,|\vecb|_0\le k}\\
		&\subseteq\Span_\F \{\deriv_{\vecz^\vecc}(FG)(\vecaa,\vecbb)\}_{|\vecc|_0\le \ell+k}
	\end{align*}
	where $\vecz=(\vecx,\vecy)$.
\end{proof}

Note that if we take $\vecaa$ and $\vecbb$ as polynomial maps in auxiliary variables $\vect$ (as will happen later), then the above lemma still holds simply by changing the field $\F$ under consideration to the field $\F(\vect)$.

This lemma shows that we can combine support-$\ell$ rank concentration on disjoint sets of $n$ variables into support-$2\ell$ rank concentration on $2n$ variables. However, this process does not change the ratio of the support to the number of variables. Nevertheless, by applying the rank condenser of \autoref{shift-rank-conc-known-basis}, we can compress back down to support-$\floor{\lg r^2}$, as we summarize in the following corollary.

\begin{corollarywp}\label{roabp-reduce}
	Let $\vecx$ and $\vecy$ be disjoint sets of $n$ distinct variables. Let $F(\vecx)\in\F[\vecx]^{r\times r}$ and $G(\vecy)\in\F[\vecy]^{r\times r}$ be of individual degree $<d$, and have support-$\ell$ rank concentration at $\vecaa,\vecbb\in\F^n$ respectively.

	Let $\vecg(\vect,\vecs)$ be an independent monomial map for total degree $<2d\ell$ mapping into the $2n$ variables $(\vecx,\vecy)$, where $u$, $\vect$ and $\vecs$ are new variables. Then $F(\vecx)G(\vecy)$ has support-$\floor{\lg r^2}$ rank concentration at $(\vecaa,\vecbb)+(u\cdot\vecg(\vect,\vecs)|_{\vecx},u\cdot\vecg(\vect,\vecs)|_{\vecy})$ over the field $\F(u,\vect,\vecs)$.
\end{corollarywp}

We now apply the above corollary recursively.  We maintain blocks of variables that have support-$\lg r^2$ rank concentration, and then merge pairs of blocks.  This raises the support to $2\lg r^2$, which is then brought back down to $\lg r^2$ with a shift.  In this entire process, the variables (and their blocks) will respect the order of the ROABP.  However, this order is only needed to be known for the analysis, as the construction is oblivious to it.  In particular, we can construct the needed monomial maps in a way oblivious to the order, as seen in \autoref{KS-mon-map-total} below.  We now prove that this scheme works.

\begin{lemma}\label{mon-map-to-gen}
	Let $n\ge 0$, $N=2^n$ and $d,r\ge 1$.  Let $\vecg:\F^m\times\F^{m'}\to\F^N$ be an independent monomial map for total degree $<2d(\floor{\lg r^2}+1)$.  Define $\cG_{n,d,r}:(\F\times\F^m\times\F^{m'})^{n}\to\F^N$ to be the polynomial map defined by
	\[
		\cG_{n,d,r}(u_1,\vect_1,\vecs_1,\ldots,u_n,\vect_n,\vecs_n)
		=\sum_{i\in[n]} u_i\cdot \vecg(\vect_i,\vecs_i).
	\]
	
	Let $\pi$ be a permutation $\pi:[N]\to[N]$.  Let $F(\vecx)=\prod_{i\in[N]} M_i(x_{\pi(i)})$, where for $i\in[N]$,  $M_i(x_{\pi(i)})\in\F[x_{\pi(i)}]^{r\times r}$ is of individual degree $<d$. Then $F(\vecx)$ has support-$(\floor{\lg r^2}+1)$ rank concentration at $\cG((u_i,\vect_i,\vecs_i)_{i\in[n]})$ over the field $\F((u_i,\vect_i,\vecs_i)_{i})$.
\end{lemma}
\begin{proof}
	We proceed by induction on $n$.

	\uline{$n=0$:} In this case, we have a single variable $x$ and the matrix $F(x)$.  The generator $\cG$ is then the zero polynomial map, as the summation defining it is empty.  Thus, $F(x)$ clearly has support-1 rank concentration at $\vec{0}$, as there is only 1 variable to consider. As $1\le \floor{\lg r^2}+1$, the claim follows.

	\uline{$n>0$:} Divide the $2^n$ variables in half, into $\vecy$ and $\vecz$, respecting the order $x_{\pi(1)},\ldots,x_{\pi(N)}$.  Thus, we can write $F(\vecx)=G(\vecy)H(\vecz)$. Now note that $\vecg(\cdot,\cdot)|_\vecy$ is a good monomial map for $\vecy$, and similarly $\vecg(\cdot,\cdot)|_\vecz$ is a good monomial map for $\vecz$.  Thus, if we apply the induction hypothesis on these cases, using their respective pieces of the monomial map $\vecg$, it follows that $G$ is support-$(\floor{\lg r^2}+1)$ rank-concentrated at $\sum_{i\in[n-1]} u_i\cdot \vecg(\vect_i,\vecs_i)|_\vecy$, and similarly $H$ is support-$(\floor{\lg r^2}+1)$ rank-concentrated at $\sum_{i\in[n-1]} u_i\cdot \vecg(\vect_i,\vecs_i)|_\vecz$. It follows then from \autoref{roabp-reduce} that $F(\vecx)=G(\vecy)H(\vecz)$ is support-$(\floor{\lg r^2}+1)$ rank concentrated at 
	\begin{align*}
		\left(\sum_{i\in[n-1]} u_i\cdot \vecg(\vect_i,\vecs_i)|_\vecy,\right.
		&\left.\sum_{i\in[n-1]} u_i\cdot \vecg(\vect_i,\vecs_i)|_\vecz\right) 
		+ \left(u_n\cdot \vecg(\vect_n,\vecs_n)|_\vecy,u_n\cdot \vecg(\vect_n,\vecs_n)|_\vecz\right)\\
		=&\left(\sum_{i\in[n]} u_i\cdot \vecg(\vect_i,\vecs_i)|_\vecy,\sum_{i\in[n]} u_i\cdot \vecg(\vect_i,\vecs_i)|_\vecz\right)=\cG_{n,d,r}((u_i,\vect_i,\vecs_i)_{i\in[n]}),
	\end{align*}
	as desired.
\end{proof}

We now instantiate the above corollary with specific monomial maps for individual degree.  We begin by noting that the SV-generator does not yield such a monomial map.  That is, for any monomial $\vecx^\veca$ of low-total-degree (for example degree $d\lg r$, as in the above lemma), the generator $\GSV{n,d\lg r}(\vecy,\vecz)$ will have some setting $\vecz\leftarrow\vecaa_\veca$ so that $\vecx^\veca$ is preserved as a monomial in $\vecy$.  However, the choice of $\vecaa_\veca$ will depend on $\veca$, and thus a different monomial $\vecx^\vecb$ of low degree may be zeroed out in the substitution $\vecz\leftarrow\vecaa_\veca$.  That is, we need that there exists a (\emph{single}) setting of $\vecz$ that preserves \emph{all} relevant monomials, whereas the SV-generator reverses the order of these quantifiers. Thus, we will instead use the KS-generator as alluded to above, as recorded in the next lemma.

\begin{lemma}\label{KS-mon-map-total}
	Let $|\F|\ge\poly(n,d)$.  Then the KS-generator $\GKS{n,d,n^D}:(\F^{\O(D\log_{nd} n)})^2\to\F^n$ of \autoref{thm:KS generator} is a $\poly(n,d,D)$-explicit, independent monomial map for total degree $<D$.
\end{lemma}
\begin{proof}
	By writing $n$-variate monomials of total degree $k$ as strings in $[n]^k$, it follows that there are $\le \sum_{k\in\zr{D}}n^k=\frac{n^D-1}{n-1}\le n^D$ monomials in $n$ variables of total degree $<D$. Thus, it follows from \autoref{thm:KS generator} that the KS-generator with sparsity $s=n^D$ will map these monomials distinctly, as desired.
\end{proof}

Combining the above, we get the desired hitting sets for ROABPs with unknown variable order, yielding the main result of this paper.

\begin{theorem}[Main]\label{hit-set-noncommut}
	Let $n\ge 0$, $N=2^n$, $d,r\ge 1$ and $|\F|\ge \poly(N,d)$. Define $\cG_{N,d,r}':(\F\times(\F^{\O(d\lg r\cdot\log_{Nd} N)})^2)^n\to\F^N$ be the polynomial map defined by
	\[
		\cG_{N,d,r}'(u_1,\vect_1,\vecs_1,\ldots,u_n,\vect_n,\vecs_n)
		=\sum_{i\in[n]} u_i\cdot \GKS{N,d,N^{2d(\floor{\lg r^2}+1)}}(\vect_i,\vecs_i).
	\]
	Then, for $i\in[N]$, the polynomial $(\cG_{N,d,r}')_i$ has individual degree $\le \poly(N,d)$. Further, $\cG_{N,d,r}'$ is $\poly(N,d,r)$ explicit.

	Further, for any $N$-variate $f(\vecx)\in\F[\vecx]$ of individual degree $<d$, computed by a ROABP of width $\le r$ in some variable order, $f\not\equiv 0$ iff $f\circ (\cG_{N,d,r}'+\GSV{\floor{\lg r^2}+1})\not\equiv 0$.

	Thus, for any $N$, there is a $\poly(N,d)^{\O(d\cdot\lg r\cdot\lg N)}$ size $\poly(N,d,r)$-explicit hitting set for width-$r$ ROABPs on $N$ variables, of individual degree $<d$, in any variable order.
\end{theorem}
\begin{proof}
	\uline{generator:} That the $\cG'+\GSV{}$ generator above is correct follows from instantiating \autoref{mon-map-to-gen} with the monomial map of \autoref{KS-mon-map-total}. 

	\uline{degree, explicitness:} These follow from \autoref{thm:KS generator}.

	\uline{hitting set:} The hitting set follows by interpolating the polynomial $f\circ(\cG'+\GSV{})$. This polynomial has individual degree $\poly(N,d)$, and has $\lg N\cdot \O(d\lg r\cdot \log_{dN} N)+\O(\lg r)$ variables. Thus, the number of points in the interpolation is $\poly(N,d)^{\O(d\lg r\cdot\lg N\cdot\log_{Nd} N)+\O(\lg r)}=N^{\O(d\lg r\cdot \lg N)}\cdot\poly(N,d)^{\O(\lg r)}=\poly(N,d)^{\O(d\lg r\cdot \lg N)}$.
\end{proof}

We now briefly comment on why our methods have a poor dependence on the individual degree of the polynomial $f$.  Based on our paradigm of merge-and-reduce, as embodied by \autoref{roabp-reduce}, we seek to move from support-$\lg r^2$ rank concentration, merge to support-$2\lg r^2$ concentration, and then reduce back to support-$\lg r^2$ concentration.  To do this reduction step, we use a rank condenser, that condenses rank from a known-basis.  This known-basis is among the monomials of support-$2\lg r^2$.  Ideally, we could construct a rank condenser that condenses just from this basis.  However, as discussed in \autoref{sec:condense-known-basis}, our rank condenser must be compatible with black-box access to polynomials.  The notion of ``support'' of a monomial is not easily compatible with black-box access, and so it unclear how to make small-support monomials both distinct, as well as ``smaller'' than all large-support monomials, such that the uniqueness arguments of \autoref{sec:condense-known-basis} (when we have that all small weights are distinct) can be applied.  

Thus, as we cannot condense from small-support monomials, the above results observe that when the individual degree $d$ is small, small-support and low-degree become nearly equivalent notions.  The notion of being low-degree is then compatible with black-box access to polynomials, via homogenization, and thus the results proceed as detailed above.

\section{Better Hitting Sets for Commutative ROABPs}\label{sec:hashing+FS}

In \autoref{sec:commutative} we gave quasipolynomial-size hitting sets for the class of commutative ROABPs.  In this section, we improve these $n^{\O(\lg n)}$ results to $n^{\O(\lg\lg n)}$, for two subclasses of commutative ROABPs: diagonal circuits, and commutative ROABPs with small individual degree.

To understand this improvement, recall the strategy of \autoref{sec:commutative}, which involved two different parts.  First, we found a small set of shifts, such that we could shift the polynomial into small-support rank concentration, and this yielded a nonzero monomial of small support.  As mentioned at the end of \autoref{sec:commutative}, we do not have good ways to reduce the number of shifts needed, and thus focus on improving the second step, which is where we try to hit a commutative ROABP that is promised to have a small support monomial.  To do this, we used \autoref{lem:GSV-hit-ss}, which composed the polynomial $f$ with the SV-generator $\GSV{n,\ell}$, where $\ell$ is the size of the small-support monomial, and observed that $f\not\equiv 0$ iff $f\circ \GSV{n,\ell}\not\equiv 0$. We then interpolated this generator into a hitting set.

If we wish to improve the hitting sets under the promise of a support-$\ell$ monomial, it seems that properties of $f$ must be used, as otherwise it seems unlikely hitting sets smaller than $\poly(n)^\ell$ can be devised as there are too many such polynomials to hit.  To take advantage of $f$, we thus ask: what is the complexity of $f\circ \GSV{n,\ell}$? Note that it has $2\ell\approx \lg r$ variables, which is quite small.  If $f\circ \GSV{n,\ell}$ itself can be expressed as a small circuit, then we could use PIT for this class.  In particular, if we have hitting sets for size-$s$, $n$-variate circuits in this class of size $\poly(s,n)^{\lg n}$, then this would yield the desired result, as $n\approx\lg r$ in this new circuit class. While this strategy can work\footnote{It can be shown that when $f$ is a (depth-3) diagonal circuit, then $f\circ \GSV{n,\ell}$ is computable by a depth-4 diagonal circuit in $\approx \lg r$ variables.  The hitting sets for depth-4 diagonal circuits of Forbes-Shpilka~\cite{ForbesShpilka12a} are of desired $\poly(s,n)^{\lg n}$ form, thus giving the $n^{\O(\lg\lg n)}$ result for depth-3 diagonal circuits.} for diagonal circuits, this section will give a different composition that work more generally for any commutative ROABP with the promise of a small-support monomial.

To hit such commutative ROABPs with a support-$\ell$ monomial, we will first hash the $n$ original variables to $\approx\ell$ buckets, and then give each bucket its own variable. Some hash function in the family will hash each variable of the support-$\ell$ monomial distinctly, which will allow us to ensure that this monomial is not canceled out, thus preserving non-zeroness.  However, we now would like to also preserve the property of being a ROABP.  Without modification, using the SV-generator will result in a polynomial that is not read-once in the variables $\vecz$.  We resolve this issue by giving each bucket its own SV-generator for support-1.  By commutativity, we can reorder the $n$ variables such that they are grouped by the buckets they hash into. Thus, each \emph{bucket} is then read-once, which implies that we now have an ROABP in $\approx \ell$ variables. We can then apply the hitting set result of Forbes-Shpilka~\cite{ForbesShpilka12a} as it has the desired $\poly(s,n)^{\lg \ell}=\poly(s,n)^{\lg\lg s}$ dependence.

We now turn to the details.  We begin by quoting the tools we will need.  We begin by quoting results showing that diagonal circuits, and the more general notion of ``low rank'' polynomials, have a small-support monomial, without any shifting.

\begin{lemma}\label{lem:rank-conc-low-rank}
	For $f(\vecx)\in\F[\vecx]$, define
	\[
		\dim(\deriv(f)) \eqdef \dim_{\F}\{\deriv_{\vecx^\veca}(f)\}_\veca
	\]
	to be the dimension of its partial derivatives, over the field $\F$.  Then $f$ is computed by a width-$\dim(\deriv(f))$, individual degree $<(\deg f+1)$ commutative ROABP.  Further, $f\not\equiv 0$ iff it contains a nonzero monomial of support $\le \lg r$. In particular, if $f$ is computed by a diagonal circuit of size $s$, then $\dim(\deriv(f))\le\poly(s,\deg(f))$.
\end{lemma}
\begin{proof}
	The structural results relating diagonal circuits, dimension of partial derivatives, and commutative ROABPs, are due to Saptharishi, as reported in Section 6 of Forbes-Shpilka~\cite{ForbesShpilka12a}.  The small-support monomial result is due to Theorem 6.9\ in Forbes-Shpilka~\cite{ForbesShpilka13}.
\end{proof}

We note that Agrawal-Saha-Saxena~\cite{AgrawalSS12} gave a set of shifts, of polynomial size, that will shift any diagonal circuit to have small-support, while the above shows that no shift is needed.

We now quote the hitting set of Forbes-Shpilka~\cite{ForbesShpilka12a} for ROABPs with variable order $x_1<\cdots<x_n$. Note that commutative ROABPs are ROABPs in \emph{any} variable order, so fixing the variable order is not a restriction in this setting.  While this hitting set is of quasipolynomial size, and the hitting sets of \autoref{sec:commutative} for commutative ROABPs are also of quasipolynomial size, the exponents of these quasipolynomials are qualitatively different in their dependencies on the parameters. The former is of the form $\poly^{\lg n}$, and the latter is always at least $d^{\lg r}$.  Thus, as we are seeking to reduce the number of variables $n$, only the former is in a position to use this improvement.

\begin{theoremwp}[Forbes-Shpilka hitting set~\cite{ForbesShpilka12a}]\label{thm:FS}
	Let $|\F|\ge \poly(n,d,r)$.  There is a $\poly(n,d,r)$-explicit, polynomial map $\GFS_{n,d,r}:\F^{\O(\lg n)}\to\F^n$ of degree $\poly(n,d,r)$, such that for any $f(\vecx)$ computed by $n$-variable, individual degree $<d$, width-$r$ ROABP, we have that $f\not\equiv 0$ iff $f\circ \GFS_{n,d,r}\not\equiv 0$.
\end{theoremwp}

Finally, we will need the usage of good hash functions, and we now define the particular property we need.

\begin{definition}
	A family $\cH=\{h:[n]\to[m]\}$ is a \emph{$(n,m,\ell)$-perfect hash function family} if for every $S\subseteq[n]$ of size $|S|\le \ell$, there is an $h\in\cH$ such that $h$ is injective on $S$.
\end{definition}

There is an extensive literature on such hash functions.  Mehlhorn~\cite{Mehlhorn82} gave matching lower and (non-constructive) upper bounds on the size of $\cH$ in such a family.  Schmidt and Siegel~\cite{SchmidtSiegel89} gave optimal constructive upper bounds in a certain regime of parameters.  However, for our range of parameters we can more simply use pairwise independent hash functions, as seen in the following standard lemma, which we instantiate with explicit pairwise independent hash functions (for example, see Carter and Wegman~\cite{CW79} or the survey of Vadhan~\cite{Vadhan12}).

\begin{lemmawp}\label{perfect-hash-via-pairwise}
	Let $\cH$ be a family of pairwise independent hash functions from $[n]\to[m]$, where $m\ge \ell^2$.  Then $\cH$ is a $(n,m,\ell)$-perfect hash function family. In particular, there is an $\poly(n,m,\ell)$-explicit $(n,m,\ell)$-perfect hash family with $m=2^{\ceil{\lg \ell^2}}$ and $|\cH|\le\poly(n,m,\ell)$.
\end{lemmawp}

We now turn to the actual construction.  One part of the construction will simply interpolate through the different hash functions in the hash family.  More interestingly, we will hash the $n$ variables into buckets, giving each bucket its own SV-generator.

\begin{construction}\label{const:hashing}
	Let $\cH$ be a $(n,m,\ell)$-perfect hash family. Let $|\F|>|\cH|,n$. Associate a distinct element $\eta_h\in\F$ for each $h\in\cH$. Define $\cGH:(\F^2)^m\times\F\to\F^n$ by
	\[
		\left(\cGH_{n,m,\ell}(y_1,\ldots,y_m,z_1,\ldots,z_m,u)\right)_i
		\eqdef \sum_{h\in\cH} \left(\GSV{n,1}(y_{h(i)},z_{h(i)})\right)_i \cdot \ind{u=\eta_h}
	\]
	for $i\in[n]$, where $\ind{w=\eta_h}$ denotes the Lagrange interpolation polynomial for the points $\{\eta_h\}_{h\in\cH}$, of degree $\le|\cH|$.
\end{construction}

In words, for every function $h\in \cH$ we assign the SV generator on variables $(y_{k},z_{k})$ to the $x$-variables that were mapped by $h$ to the $k$-th bucket. We now give the desired properties of this construction.

\begin{lemma}\label{lem:hashing}
	Assume the setup of \autoref{const:hashing}, and the setup of the SV-generator of \autoref{thm:SV generator}. Let $f(\vecx)\in\F[\vecx]$ be computed by a width-$r$, individual degree $<d$, commutative ROABP.  Suppose that $f$ has a nonzero support-$(\le \ell)$ monomial.  Then there is a value of $\alpha\in\F$ such that $f\circ \cGH_{n,m,\ell}(\vecy,\vecy^{dn^2},\alpha)$ is a nonzero width-$r$, individual degree $<d^2n^4$, commutative ROABP in $m$ variables.
\end{lemma}
\begin{proof}
	\uline{$f\circ \cGH\not\equiv 0$:} Let the nonzero support-$(\le\ell)$ monomial be $\vecx^\veca$.  Let $S=\supp(\veca)\subseteq[n]$, thus $|S|\le\ell$.  It follows that there is some hash $h\in\cH$ such that $h$ is injective on $S$.  Let $\alpha\leftarrow\eta_h$.  Thus, we have that $(\cGH_{n,m,\ell}(\vecy,\vecz,\alpha))_i=(\GSV{n,1}(y_{h(i)},z_{h(i)}))_i$.  For $j\in[m]$, set $z_j\leftarrow\xi_{h^{-1}(j)\cap S}$, where $h^{-1}(j)\cap S$ is the unique element in $S$ that hashes to $j$, or otherwise we take it to be $0$. Call this assignment $\vecz\leftarrow\vec{\xi}_0$.
	
	Recalling that taking $z_j\leftarrow \xi_0$ zeroes out the corresponding $y_j$, and that taking $z_j\leftarrow \xi_k$ makes $y_j$ only appear in the $k$-th output slot, we see that $\cGH_{n,m,\ell}(\vecy,\vec{\xi}_0,\alpha))$ places distinct variables of $\vecy$ into the set $\supp(\vecx^\veca)$, and otherwise outputs 0.  It follows that $f\circ \cGH_{n,m,\ell}(\vecy,\vec{\xi}_0,\alpha))$ contains exactly those monomials in $f$ whose support is contained in $\supp(\vecx^\veca)$, and there is no cancellation.  Therefore, the nonzero monomial $\vecx^\veca$ survives this assignment, thus $f\circ \cGH_{n,m,\ell}(\vecy,\vec{\xi}_0,\alpha))\not\equiv 0$ and hence $f\circ \cGH_{n,m,\ell}(\vecy,\vecz,\alpha))\not\equiv 0$.

	Now note that $f$ has individual degree $<d$, and thus total degree $<dn$. The generator $\cGH_{n,m,\ell}$ has individual degree $\le n$ in $\vecy$ and $\vecz$.  It follows that $f\circ \cGH_{n,m,\ell}(\vecy,\vecz,\alpha)$ has individual degree $<dn^2$ in each variable in $\vecy$ and $\vecz$.  Thus, applying the Kronecker substitution it follows that we can substitute $\vecz\leftarrow\vecy^{dn^2}$ and preserve nonzero-ness.

	\uline{individual deg $<d^2n^4$:} As argued above, before the substitution ``$\vecz\leftarrow\vecy^{dn^2}$'' we had individual degree $<dn^2$, and the substitution multiplies this by another $dn^2$ factor.

	\uline{$f\circ \cGH$ has a small commutative ROABP:} Let $\pi:[m]\to[m]$ be any permutation, so that we wish to show that $f\circ \cGH$ has a small ROABP in the variable order $y_{\pi(1)},\ldots, y_{\pi(m)}$. Note that $\pi$ induces an partial order $\prec$ on $[n]$ via the hash function $h$ chosen above.  That is, say that for $i,j\in[n]$ that $i\prec j$ if $\pi^{-1}(h(i))<\pi^{-1}(h(j))$. Now extend $\prec$ to a total order on $[n]$ arbitrarily, and also call the induced permutation $\sigma:[n]\to[n]$. Thus, $\sigma$ is the permutation that orders the variables according to $\prec$.

	As $f$ has a commutative ROABP, it follows that $f(\vecx)$ is a linear function of the entries of a matrix-valued ROABP $F(\vecx)=\prod_{i\in[n]} M_i(x_{\sigma(i)})$ for $M_i(x_{\sigma(i)})\in \F[x_{\sigma(i)}]^{r\times r}$ of individual degree $<d$. Now we plug in the above generator, and this yields that $f\circ \cGH_{n,m,\ell}(\vecy,\vecy^{dn^2},\alpha)$ is a linear combination of the entries of
	\begin{align*}
		\prod_{i\in[n]} M_i\left( \left(\cGH_{n,m,\ell}(\vecy,\vecy^{dn^2},\alpha)\right)_{\sigma(i)}\right)
		&=\prod_{i\in[n]} M_i\left( \GSV{n,1}\left(y_{h(\sigma(i))},y^{dn^2}_{h(\sigma(i))}\right)\right)\\
		&=\prod_{j\in[m]} \prod_{h(\sigma(i))=\pi(j)} M_i\left( \GSV{n,1}\left(y_{\pi(j)},y_{\pi(j)}^{dn^2}\right)\right)\\
		&=\prod_{j\in[m]} M_j'(y_{\pi(j)})
	\end{align*}
	where the above products are always increasing in $i\in[n]$ and $j\in[m]$ (so that this respects the non-commutativity of this matrix product, via the choice of $\pi$ and $\sigma$), and 
	\[
		M_j'(y_{\pi(j)})\eqdef \prod_{h(\sigma(i))=\pi(j)} M_i\left( \GSV{n,1}\left(y_{\pi(j)},y_{\pi(j)}^{dn^2}\right)\right)
	\]
	so that $M_j'(y_{\pi(j)})\in\F[y_{\pi(j)}]^{r\times r}$.  It follows that we have expressed $f\circ \cGH$ as a width-$r$ ROABP in the order defined by $\pi$, and as $\pi$ was arbitrary, it follows that $f\circ \cGH$ has a width-$r$ commutative ROABP.
\end{proof}

Thus, we have shown that commutative ROABPs with small-support monomials can have their number of variables essentially reduced to the size of that small-support monomial.  We now combine this with the results of Forbes-Shpilka~\cite{ForbesShpilka12a} to derive our improved hitting sets.

\begin{theorem}\label{nlglgn}
	Let $|\F|\ge\poly(n,d,r,\ell)$. There is a $\poly(n,d,r,\ell)$-explicit, polynomial map $\cGHFS_{n,d,r,\ell}:\F^{\O(\lg\ell)}\to\F^n$ of degree $\poly(n,d,r,\ell)$, such that for any $f(\vecx)\in\F[\vecx]$ computed by a width-$r$, individual degree $<d$ commutative ROABP, where $f\not\equiv 0$ iff $f$ has a nonzero support-$(\le\ell)$ monomial, we have that $f\not\equiv 0$ iff $f\circ \cGHFS_{n,d,r,\ell}\not\equiv 0$.

	In particular, when $|\F|\ge\poly(n,d,s)$, the class of size-$s$, $n$-variable, total degree $d$ diagonal circuits has $\poly(n,d,s)$-explicit hitting sets of size $\poly(n,d,s)^{\O(\lg\lg(ds))}$.  More generally, when $|\F|\ge\poly(n,r)$, the class of $n$-variate polynomials $f$ with $\dim(\deriv f)\le r$ has $\poly(n,r)$-explicit hitting sets of size $\poly(n,r)^{\O(\lg\lg r)}$.

	Similarly, if $|\F|\ge \poly(n,d,r)$, the class of width-$r$, $n$-variable, individual degree $<d$ commutative ROABPs has a $\poly(n,d,r)$-explicit hitting set of size $\poly(n,d^{\lg r})\cdot\poly(n,d,r)^{\O(\lg\lg r)}$.
\end{theorem}
\begin{proof}
	\uline{$\impliedby:$} Clear.

	\uline{$\implies:$} By assumption, $f\not\equiv0$ and has a support-$(\le\ell)$ monomial.  We can instantiate the hashing generator $\cGH$ of \autoref{lem:hashing} with the perfect hashing family of \autoref{perfect-hash-via-pairwise}, so that $f\circ \cGH_{n,\O(\ell^2),\ell}$ is nonzero and computable by a width-$r$, individual degree $\poly(d,n)$, commutative ROABP on $\O(\ell^2)$ variables, for some value of the variable $u$. It thus follows from \autoref{thm:FS} that $f\circ \cGH_{n,\O(\ell^2),\ell}\circ\GFS_{\O(\ell^2),\poly(n,d),r}\not\equiv 0$, as this was true for some value of $u$, so is also true when $u$ is symbolic. Taking $\cGHFS\eqdef \cGH\circ \GFS$ with the relevant parameters, we see that this generator has the desired explicitness, degree bound, and number of variables.

	\uline{diagonal/low-rank circuits:} \autoref{lem:rank-conc-low-rank} recalled that nonzero low-rank polynomials have, without shifting, a nonzero monomial of support $\O(\lg(sd))$ (or $\O(\lg r)$ for low-rank circuits), and are computed by commutative ROABPs of width-$\poly(s,d)$ (or width-$\poly(r)$ for low-rank circuits), so applying the above generator and interpolating gives the result.

	\uline{low individual degree commutative ROABPs:} Given such an $f$, we can shift it to have support-$\O(\lg r)$ rank concentration, by instantiating \autoref{thm:rank-concentration-comm} with the $\O(\lg r)$-wise independent monomial map from the KS-generator given in \autoref{mon-map-KS}.  This $\poly(n,d,r)$-explicit monomial map has individual degree $\poly(n,d)$, and $\O(\lg r\cdot\log_{nd} d)$ variables.  As this gives a small-support monomial, we can then hit $f$ by adding in the $\cGHFS$ generator above, and then interpolate the entire polynomial.  The KS-generator will use $\poly(n,d^{\lg r})$ possible values, and this will multiply on the number of values from the $\cGHFS$ generator.
\end{proof}

Note that the last result on commutative ROABPs becomes $(nr)^{\O(\lg\lg r)}$ when $d=\O(1)$.

\section*{Acknowledgements}

This work was done while the first two authors were visiting the third.  The first two authors would like to thank Chandan Saha for explaining the work of Agrawal-Saha-Saxena~\cite{AgrawalSS12} to them.

\bibliographystyle{alphaurl}
\bibliography{bibliography}

\end{document}